\documentclass[11pt]{elsarticle}
\usepackage{epsfig,amsfonts,amssymb,amsmath,amsthm}
\usepackage{color}
\usepackage{array}
\usepackage{lineno}
\usepackage{multirow}
\usepackage{graphicx}
\usepackage{subcaption}
\usepackage{bm}

\makeatletter
\def\ps@pprintTitle{%
  \let\@oddhead\@empty
  \let\@evenhead\@empty
  \let\@oddfoot\@empty
  \let\@evenfoot\@oddfoot
}

\newtheorem{thm}{Theorem}[section]

\newtheorem{lem}[thm]{Lemma}

\renewcommand{\epsilon}{\varepsilon}

\makeatother

\begin{document}

\begin{frontmatter}

\title{Saddle-node bifurcation of limit cycles in an epidemic model with two levels of awareness}

\author[address1]{David Juher}
\ead{david.juher@udg.edu}

\author[address1]{David Rojas}
\ead{david.rojas@udg.edu}

\author[address1]{Joan Salda\~{n}a\corref{cor1}}
\ead{joan.saldana@udg.edu}

\address[address1]{Departament d'Inform\`{a}tica, Matem\`{a}tica Aplicada i Estad\'istica
\\
Universitat de Girona, Girona 17003, Catalonia, Spain}

\cortext[cor1]{Corresponding author}

\begin{abstract}
In this paper we study the appearance of bifurcations of limit cycles in an epidemic model with two types of aware individuals. All the transition rates are constant except for the alerting decay rate of the most aware individuals and the rate of creation of the less aware individuals, which depend on the disease prevalence in a non-linear way. For the ODE model, the numerical computation of the limit cycles and the study of their stability are made by means of the Poincaré map. Moreover, sufficient conditions for the existence of an endemic equilibrium are also obtained. These conditions involve a rather natural relationship between the transmissibility of the disease and that of awareness. Finally, stochastic simulations of the model under a very low rate of imported cases are used to confirm the scenarios of bistability (endemic equilibrium and limit cycle) observed in the solutions of the ODE model.   

\end{abstract}

\begin{keyword}
epidemic models, awareness, bifurcations, limit cycles, stochastic simulations.
\end{keyword}

\end{frontmatter}



\section{Introduction}\label{intro}

The role of human behaviour has been increasingly considered in epidemiological modelling since the early 2000s \cite{FSJ}. The spread of COVID-19 has highlighted even more its important role in the progress of infectious diseases. Besides institutional measures as mobility restrictions, mandatory use of facemasks, or school closings, self-initiated individual behaviours related to risk aversion are recognized as a driving force in epidemic dynamics \cite{Manrubia, Weitz}.  

One way to model such behavioural changes in deterministic models is to modify the incidence term $\beta S I$ where $\beta$ denotes the rate of disease transmission, and $S$ and $I$ are the number of susceptible and infected individuals, respectively. The simplest way to modify it is by assuming that $\beta$ is no longer constant but a decreasing function of the prevalence of the disease (\cite{Capasso, Ruan, Alexander, Weitz}). In this mean-field formulation of the incidence term, $\beta$ depends on the contact rate as well as on the probability of transmission during an infectious contact. So, its reduction can reflect a diminution in the number of social contacts (social distancing), the adoption of measures to prevent infection while keeping the same contact rate (decrease of the infection probability), or both.

On the other hand, it is well known that the perception of infection risk is uneven among susceptible individuals \cite{Guenther}. One way to introduce some heterogeneity in risk-taking propensity has been to include more types of uninfected individuals characterised by their level of responsiveness to risk. For instance, the Susceptible-Aware-Infectious-Susceptible (SAIS) model considers a new class of non-infected individuals with a higher risk aversion than the susceptible ones, the so-called aware or alerted individuals, who are characterized by a lower transmission rate \cite{Sahneh12a}. 

A basic ingredient in such a modelling approach is the transmission of awareness among individuals \cite{Funk09}. In \cite{JSX} the authors considered an SAIS model where alerted individuals were able to transmit awareness by convincing non-aware individuals to take preventive measures against the infection, which is an example of self-initiated individual behaviour. Moreover, a new class of aware individuals, the so-called unwilling (U) individuals, is also introduced. They are characterized by a lower level of alertness which is translated into a lack of willingness to transmit awareness to susceptible individuals. The existence of this second class of aware individuals turns out to be necessary to have oscillatory solutions of the SAUIS model with no births and deaths in the population.   

The inflow of new susceptible individuals in the population is a key factor in mean-field epidemic models to observe periodic solutions \cite{Ruan}. In dynamic networks models, link dynamics can also play this role \cite{Szabo}. However, even without demographic processes, behaviourally-induced epidemic oscillations can also be expected to occur when individuals experience a decline in awareness as a result of preventive measures taken over long periods of time combined with low disease prevalence. This fact was, indeed, proved in \cite{JSX} by analysing the occurrence of a Hopf bifurcation from an endemic equilibrium of the SAUIS model. Later, the robustness of such oscillations was confirmed in \cite{JRS2} under the assumption of a low rate $\epsilon$ of imported cases (infections contracted from abroad) by means of stochastic simulations on random networks. 

In this paper, we explore an extended version of the SAUIS-$\epsilon$ model in \cite{JRS2} which considers that awareness dynamics changes abruptly when disease prevalence crosses a threshold value $\eta$. Precisely, the rate $\nu_a$ of creation of unwilling individuals and the rate of awareness decay $\delta_a$ are modulated by the following function $\sigma_m(i)$ of the fraction $i$ of infected individuals:
$$
\sigma_m(i) = \frac{1}{1+(i/\eta)^m}, \quad m \ge 1.
$$ 
The sharpness of the reduction of these two rates is controled by the parameter $m$, while $\eta$ is the half-saturation constant (see Figure \ref{fig:sigma}). In particular, for $m \gg 1$, $1-\sigma_m(i)$ becomes closer to the unit step function $\theta(i-\eta)$. An extreme example of such abrupt behavioural responses could be the occurrence of panic waves when new cases of an emerging disease appear in a population.  

In contrast to other papers where awareness is considered in terms of non-constant infection transmission rates (see, for instance, \cite{Alexander, Capasso, Liu, Ruan, Weitz}), here we will focus on the awareness dynamics themselves and their role in the appearance of periodic solutions (oscillatory epidemics). In particular, we are interested in how the behaviour of solutions is affected by the reduction of both the decay of awareness and the creation of unwilling individuals.   

\section{SAUIS-$\epsilon$ model with varying coefficients}

Each individual in a population can be in one of the following four states: S (susceptible), A (aware), U (unwilling), and I (infected). The model assumes that aware individuals are created at alerting rates $\alpha_i$ and $\alpha_a$ from susceptible ones after being in contact with infected and aware individuals, respectively. Aware individuals experience an alerting decay and become unwilling at a rate $\delta_a$, while unwilling individuals also appear at rate $\nu_a$ from contacts between susceptibles and aware individuals (nodes) and they become susceptible at a rate $\delta_u$. The infection transmission rates for susceptible, aware, and unwilling individuals are $\beta$, $\beta_a$, and $\beta_u$, respectively, while the recovery rate from infection is $\delta$.     

Moreover, following the SAUIS-$\epsilon$ model introduced in \cite{JRS2}, we consider the arrival of imported cases (infections contracted abroad) at a very low rate $\epsilon > 0$. This fact prevents stochastic epidemic oscillations from extinction and, at the same time, the dynamical properties of the solutions remain close to those of the deterministic ODE model with $\epsilon=0$. 

Finally, as explained in the introduction, we assume that the rate of awereness decay $\delta_a$ and the rate $\nu_a$ at which susceptible hosts become unwilling due to a contact with an aware host both depend on the fraction of infected individuals through a reduction factor given by the function $\sigma_m(i)$. The resulting ODE system governing the epidemic dynamics is then given by:
\begin{equation}\label{eqn:SAUIS}
\begin{split}
\frac{da}{dt} & =  \alpha_i \, s \, i + \alpha_a \, s \, a  -  \beta_a a \, i - \delta_a\,\sigma_m(i)\,a - \epsilon a, \quad \beta_a < \beta, \\
\frac{du}{dt} & =  \delta_a\,\sigma_m(i)\,a  + \nu_a\,\sigma_m(i)\, s \, a -  \beta_u u \, i - \delta_u u - \epsilon u, \quad \beta_u < \beta, \\
\frac{di}{dt} & =  (\beta  \, s + \beta_a a + \beta_u u - \delta) i + (1-i)\epsilon, \qquad s+a+u+i = 1,
\end{split}
\end{equation}
where $s$, $a$, $u$, and $i$ denote the fractions of hosts in the $S$, $A$, $U$, and $I$ compartments, respectively. The differential equation for $s$ has been omitted because it is redundant. 

\begin{figure}
    \centering
    \includegraphics[scale=0.55]{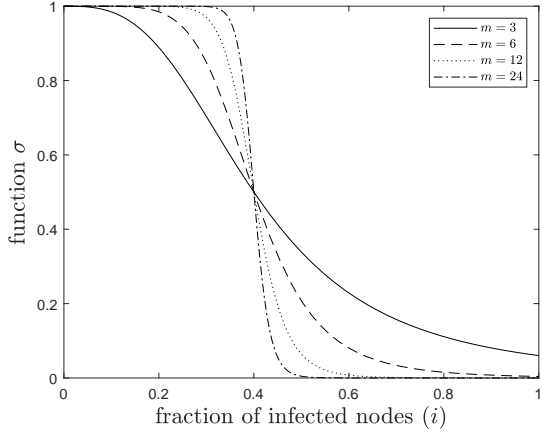}
    \caption{Shape of function $\sigma_m(i)$ with $\eta=0.4$ for different values of $m$.}
    \label{fig:sigma}
\end{figure}

\section{Equilibria}

The natural state space of system \eqref{eqn:SAUIS} is $\Omega := \{ (a,u,i) \in \mathbb{R}^3 : 0 \le a+u+i \le 1 \}$. The existence of imported cases from abroad guarantees that the vector field defined by this system on the boundary of $\Omega$ points strictly towards its interior. In particular, this implies that $\Omega$ is positively invariant under the flow defined by the solutions of system \eqref{eqn:SAUIS} and, moreover, the non-existence of disease-free equilibria for this model. 

On the other hand, since $\sigma_m(0)=1$, the same analysis of the bifurcations from the two disease-free equilibria (DFE) of the model with $\epsilon=0$ done in \cite{JRS2} works for our system. For instance, taking $\epsilon$ as a bifurcation parameter, it follows that one interior equilibrium of \eqref{eqn:SAUIS} comes from the bifurcation of a DFE of the system with $\epsilon=0$ when $\beta < \delta$. Precisely, either the DFE $\mathbf{e}^*_1=(0,0,0)$ enters $\Omega$ for $\epsilon>0$ if $\alpha_a < \delta_a$, or the DFE $\mathbf{e}^*_2=(a^*_0,u^*_0,0)$ with $a^*_0 = \delta_u\left(1-\delta_a/\alpha_a\right)/\left(\delta_a\left(1+\nu_a/\alpha_a\right)+\delta_u\right)$ and $u^*_0=\delta_a/\delta_u\left(1+\nu_a/\alpha_a\right)a^*_0$ 
enters $\Omega$ for $\epsilon>0$ if $\alpha_a > \delta_a$. In both cases, the interior equilibrium of system \eqref{eqn:SAUIS} bifurcates from an asymptotically stable DFE and is only maintained by the presence of imported cases. So, such an equilibrium is not a proper endemic equilibrium because it does not result from the disease transmission within the population.

For $\beta > \delta$ and taking $\beta_a$ as a bifurcation parameter, it follows that $\mathbf{e}^*_2$ is still asymptotically stable if $\beta_a < \beta_a^c:=\beta-\left(\beta-\delta-(\beta-\beta_u)u^*_0\right)/a^*_0$. In this case, an interior equilibrium fed by the imported cases bifurcates from it. So, from now on we will assume that $\beta > \delta$ and $\beta_a>\beta_a^c$ to guarantee that no interior equilibrium for $\epsilon > 0$ arises from a DFE with $\epsilon = 0$ and, hence, that any interior equilibrium corresponds to the perturbation of an endemic equilibrium of the system with $\epsilon=0$. 

Endemic equilibria are, in general, very difficult to determine analytically. When $\epsilon=0$ we can easily see that any endemic equilibrium lies inside the plane
\begin{equation}\label{plane}
1-\frac{\delta}{\beta}-\left(1-\frac{\beta_a}{\beta}\right)a-\left(1-\frac{\beta_u}{\beta}\right)u - i = 0.
\end{equation}
The following result gives sufficient conditions for the existence of at least one endemic equilibrium point of the model~\eqref{eqn:SAUIS} with $\epsilon=0$. The proof relies on a version of the Poincaré-Miranda theorem in a triangular domain, which we include in the Appendix for completeness.

\begin{lem}\label{equilibrium}
System~\eqref{eqn:SAUIS} with $\epsilon=0$ has an endemic equilibrium point in $\Omega$ if $0\leq \beta_a<\beta_u<\delta<\beta$ and $\frac{\alpha_a}{\delta_a}<\frac{\beta-\beta_a}{\delta-\beta_a}$.
\end{lem}
\begin{proof}
Substituting~\eqref{plane} into the first and second equations of~\eqref{eqn:SAUIS} we obtain two continuous functions in the variables $(a,u)$, $f_1$ and $f_2$, respectively. We find endemic equilibria in the common zeros of $f_1(a,u)$ and $f_2(a,u)$. The intersection of the plane~\eqref{plane} with $\Omega$ projected to the $(a,u)$-plane is the right triangle with vertex $(0,0)$, $(\tfrac{\beta-\delta}{\beta-\beta_a},0)$ and $(0,\tfrac{\beta-\delta}{\beta-\beta_u})$. Notice that $\tfrac{\beta-\delta}{\beta-\beta_a},\tfrac{\beta-\delta}{\beta-\beta_u}<1$ since $\beta_u,\beta_a<\delta$. The hypothenusa of the triangle is given by substituting $i=0$ in the equation~\eqref{plane}.

On the one hand, $f_1(0,u)=\frac{\alpha_i\beta_u(\beta-\beta_u)}{\beta^2}u^2-\frac{\alpha_i(\beta\beta_u+\beta\delta-2\beta_u\delta)}{\beta^2}u+\frac{\alpha_i\delta(\beta-\delta)}{\beta^2}$ is positive for $u\in[0,\frac{\beta-\delta}{\beta-\beta_u})$ and vanishes at $u=\frac{\beta-\delta}{\beta-\beta_u}$. On the hypothenusa, $f_1$ takes the value
\[
f_1(a,\tfrac{\beta-\delta-(\beta-\beta_a)a}{\beta-\beta_u})=\frac{a}{\beta-\beta_u}(\alpha_a(\beta_u-\beta_a)a+\alpha_a(\delta-\beta_u)-\delta_a(\beta-\beta_u))
\]
which vanishes at $a=0$ and $a=\frac{\delta_a(\beta-\beta_u)-\alpha_a(\delta-\beta_u)}{\alpha_a(\beta_u-\beta_a)}$. Elementary computations show that this second root is greater than $\frac{\beta-\delta}{\beta-\beta_a}$ if and only if
\[
(\beta-\beta_u)(\delta_a(\beta-\beta_a)-\alpha_a(\delta-\beta_a))>0,
\]
which follows from the hypotheses. Thus $f_1<0$ on the hypothenusa (note that the coefficient of $a^2$ is strictly positive since $\beta_u > \beta_a$).

On the other hand, $f_2(0,u)=\frac{\beta_u(\beta-\beta_u)}{\beta}u^2-\frac{(\beta-\delta)\beta_u+\beta\delta_u}{\beta}u$ is negative for all $u\in(0,\frac{\beta-\delta}{\beta-\beta_u})$ and vanishes at $u=0$, and 
\[
f_2(a,0)= \frac{\sigma_m(i)}{\beta} \left(-\beta_a\nu_a a^2+(\beta\delta_a+\delta\nu_a)a\right)
\]
vanishes at $a=0$ and $a=\frac{\beta\delta_a+\delta\nu_a}{\beta_a\nu_a}>1$, so $f_2(a,0)>0$ for all $a\in(0,\frac{\beta-\delta}{\beta-\beta_a})$. 

Therefore, $f(a,u):=(f_1(a,u),f_2(a,u))$ satisfies the assumptions of Theorem~\ref{thm:poincare-miranda} (see Appendix) and there exists a common zero of $f_1(a,u)$ and $f_2(a,u)$ inside the triangle formed by the plane~\eqref{plane} inside $\Omega$, which corresponds to an endemic equilibrium.
\end{proof}

Note that, if $\frac{\alpha_a}{\delta_a} < \frac{\beta}{\delta}$ and the first hypothesis of the lemma is fulfilled, then the second one is guaranteed because $g(x)=\frac{\beta-x}{\delta-x}$ is an increasing funtion of $x$ and, hence, we have
$$
\frac{\alpha_a}{\delta_a} < \frac{\beta}{\delta} < \frac{\beta-\beta_a}{\delta-\beta_a}. 
$$
In other words, if the basic reproduction number of the disease in an awareness-free population, $\beta/\delta$, is larger than $1$ and, also, is larger than that of the awareness transmission in a wholly susceptible population, $\alpha_a/\delta_a$, then the existence of an endemic equilibrium for $\epsilon = 0$ is guaranteed. 

With respect to the endemic equilibria of system~\eqref{eqn:SAUIS} for $\epsilon>0$, we know that any equilibrium whose existence is guaranteed by Lemma~\ref{equilibrium} will persist inside $\Omega$ for $\epsilon>0$ small enough under the classical transversal condition thanks to the implicit function theorem.

\section{Saddle-node bifurcation of limit cycles}

A bifurcation that passes from the exhibition of two hyperbolic limit cycles of different stability in the phase portrait to the absence of such limit cycles through their collision in a non-hyperbolic semistable limit cycle is called a \emph{saddle-node bifurcation of limit cycles}. A classic scenario in the plane where this bifurcation appears is through a bistable configuration, when a stable equilibrium point is surrounded by an also stable limit cycle. Therefore a second limit cycle, in this case unstable, separates the basins of attraction of both stable objects. A continuous dependence on parameters may cause the two limit cycles collide and initial solutions that were attracted to the oscillatory motion are then attracted to the equilibrium. This bifurcation is not exclusive of the plane, although for greater dimensions the existence of the equilibrium is not required.

In order to find limit cycles and to determine their stability, the \emph{Poincaré map} is used. In three-dimensional vector fields, as it is the case for the model under consideration, the Poincaré map (also known as \emph{first-return map}) is a two-dimensional discrete map from a plane transversal to the flow of the system located near the periodic orbit to itself. The image by the Poincaré map of each point on the transversal plane is the next intersection point of the flow on the plane. Fixed points of the Poincaré map correspond to limit cycles and the stability of such fixed points gives the stability of the periodic orbit. 

\subsection{Numerical computation of limit cycles and its stability}

Analytic treatment of bifurcations involving limit cycles are only available when the expression of the periodic orbit is known as a function of the parameter producing the bifurcation. Since such expression is usually not computable in applications, numerical methods are the common technique for the detection of saddle-node bifurcations of limit cycles. 

Consider a system of differential equations  $\dot{\bm{x}}=\bm{f}(\bm{x})$ for $\bm{x}\in\mathbb{R}^n$ and the associated flow $\bm{\varphi}_t(\bm{x})$, $t \in \mathbb{R}$. Consider a hypersurface $\Sigma$ of $\mathbb{R}^n$ transversal to the vector field and assume there exists $\bm{x}_0\in\Sigma$ such that $\bm{\varphi}_{T_0}(\bm{x}_0)\in\Sigma$ for some minimal time $T_0>0$. Continuous dependence on initial conditions of the system provides the existence of a neighbourhood $U$ of $\bm{x}_0$ and a function $T:U\rightarrow \mathbb{R}$ such that $\bm{\varphi}_{T(\bm{x})}(\bm{x})\in\Sigma$ for all $\bm{x}\in U$, the so-called time-return map. The map $\bm{P}:\Sigma\rightarrow\Sigma$ defined by
$\bm{P}(\bm{x}) = \bm{\varphi}_{T(\bm{x})}(\bm{x})$
is the Poincaré map or first-return map of the section $\Sigma$.

Fixed points of $\bm{P}$ correspond to limit cycles of the system. So, in order to locate them, the Newton-Raphson's method can be applied to the distance function $\bm{F}(\bm{x})=\bm{P}(\bm{x})-\bm{x}$ near the limit cycle we wish to locate. The iterative procedure
\[
\bm{x}_{n+1} = \bm{x}_n- D\bm{F}(\bm{x}_n)^{-1}\bm{F}(\bm{x}_n),
\]
where $D\bm{F}$ denotes the Jacobian matrix of $\bm{F}$ and $\bm{x}_n\in\Sigma$, produces successively better approximations of the periodic orbit. Rather than computing the inverse of the Jacobian matrix, the usual and more numerically stable procedure is to solve the linear system
\begin{equation}\label{newton}
D\bm{F}(\bm{x}_n)(\bm{x}_{n+1}-\bm{x}_n)=-\bm{F}(\bm{x}_n).
\end{equation}
Newton-Raphson's method is an effective way to find isolated periodic orbits as long as the Jacobian matrix $D\bm{F}$ is non-singular. In order to integrate the solution we use Runge-Kutta's method RK45. At the same time the solution $\bm{\varphi}_t(\bm{x})$ is integrated, we also integrate the variational equation
\[
\dot{\bm{Y}} = D\bm{f}(\bm{\varphi}_t(\bm{x}))\bm{Y}, \; \bm{Y}(0)=I_n
\]
obtaining the monodromy matrix $D\bm{\varphi}_{T(\bm{x})}(\bm{x})$. Therefore the differential of the Poincaré map can be computed as
\[
D\bm{P}(\bm{x}) = \frac{d}{d t} \left( \varphi_{T(x)}(\bm{x}) \right) + D\varphi_{T(\bm{x})}(\bm{x}) = \bm{f}(\bm{P}(\bm{x})) DT(\bm{x}) + D\varphi_{T(\bm{x})}(\bm{x}).
\]
If the hypersurface $\Sigma$ is defined by $\{g(\bm{x})=0\}$ and it is traversed from $\{g(\bm{x})<0\}$ to $\{g(\bm{x})>0\}$, by implicit derivation of $g(\bm{x})=0$ we can find the differential of the time-return map
\[
DT(\bm{x})=-\frac{Dg(\bm{P}(\bm{x}))D\bm{\varphi}_{T(\bm{x})}(\bm{x})}{Dg(\bm{P}(\bm{x}))\bm{f}(\bm{P}(\bm{x}))}.
\]
Then the Jacobian matrix of the Poincaré map can be written as
\[
D\bm{P}(\bm{x}) = -\bm{f}(\bm{P}(\bm{x}))\frac{Dg(\bm{P}(\bm{x}))D\bm{\varphi}_{T(\bm{x})}(\bm{x})}{Dg(\bm{P}(\bm{x}))\bm{f}(\bm{P}(\bm{x}))}+D\bm{\varphi}_{T(\bm{x})}(\bm{x}).
\]
Finally, $D\bm{F}=D\bm{P}-I$ and the iterative procedure \eqref{newton} can be used to locate limit cycles (we refer the reader to \cite{Mondelo} for more details). Once the limit cycle is located, the eigenvalues of the monodromy matrix $D\bm{\varphi}_{T(\bm{x})}(\bm{x})$ at the limit cycle, the so-called Floquet characteristic multipliers, give the stability of the limit cycle found.

\begin{figure}
\begin{subfigure}{.49\textwidth}
    \centering
    \includegraphics[scale=0.5]{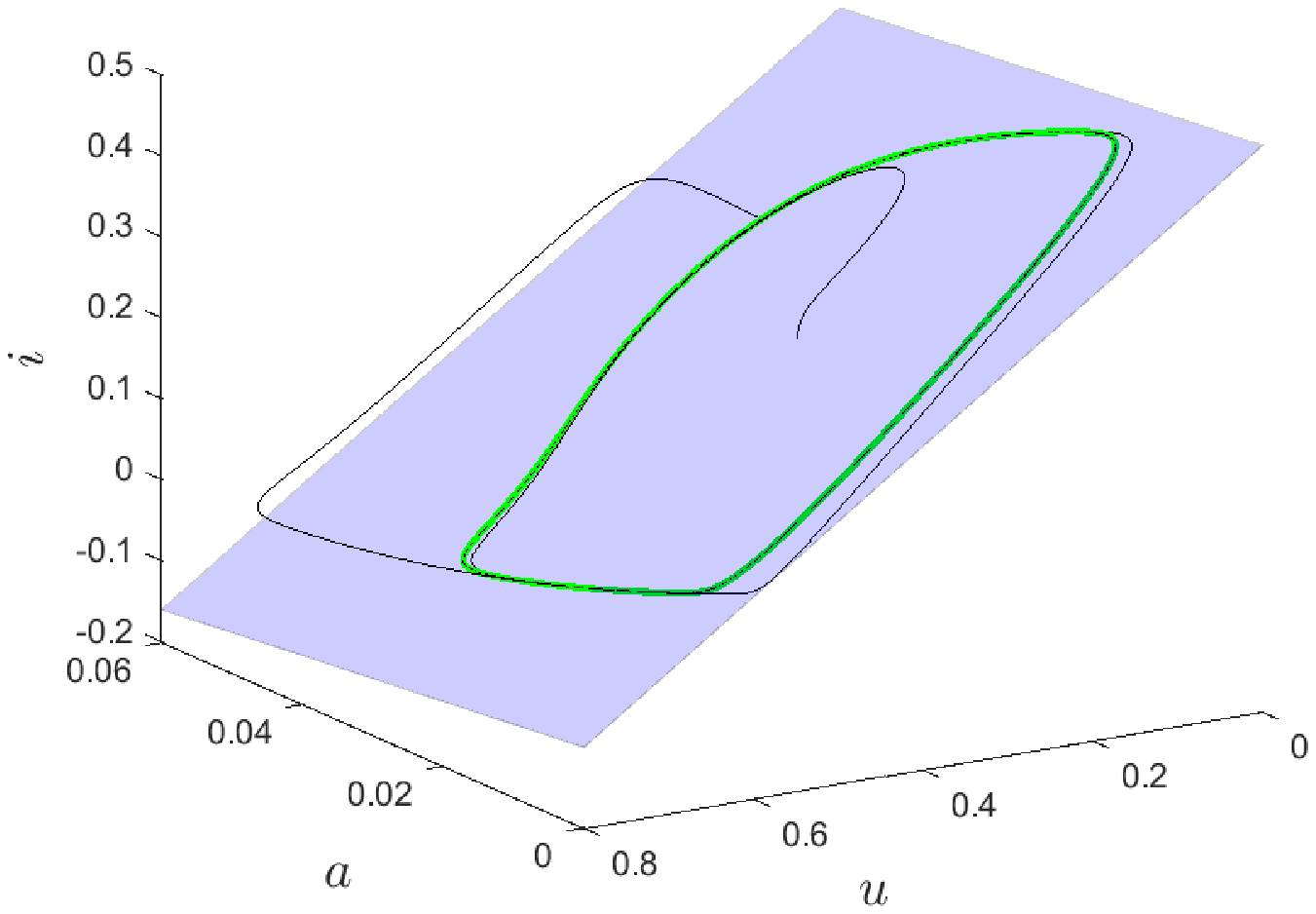}
\end{subfigure}
\begin{subfigure}{.49\textwidth}
    \centering
    \includegraphics[scale=0.5]{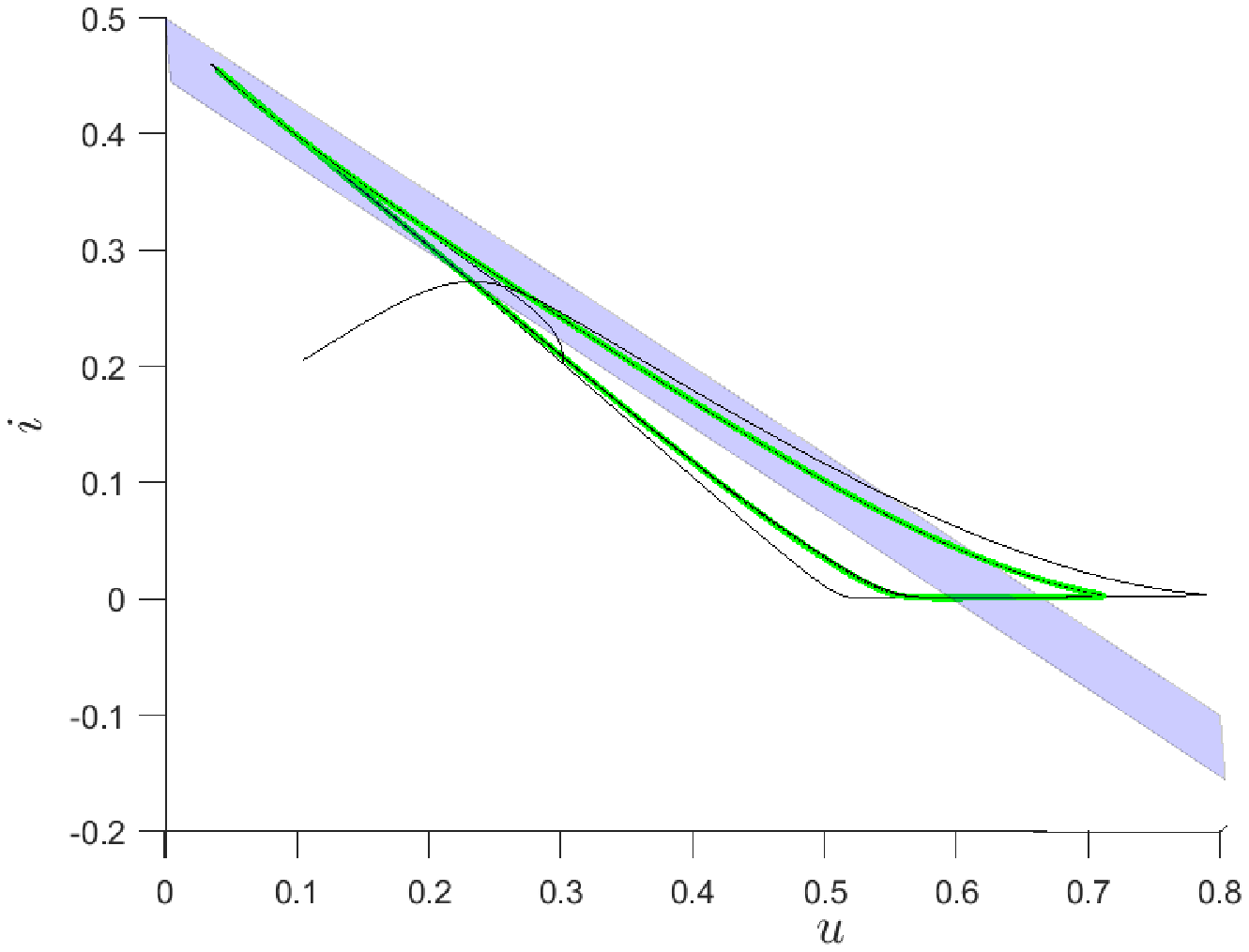}
\end{subfigure}
\caption{Two solutions of system~\eqref{eqn:SAUIS} for $m=3.5$ tending to the limit cycle (in green) whose dynamics is close to the plane~\eqref{plane} (in light blue).}
\label{fig:orbit_plane}
\end{figure}

\subsection{The saddle-node bifurcation of limit cycles in the model}

The version of the SAUIS-$\epsilon$ model~\eqref{eqn:SAUIS} exhibits a saddle-node bifurcation of limit cycles in a certain region of the parameter space using $m$ as a bifurcation parameter. All along the paper we use the following values for the parameters, which we consider fixed: $\beta=2$, $\beta_a=0$, $\beta_u=0.5$, $\alpha_a=0.015$, $\alpha_i=0.001$, $\nu_{a}=3$, $\delta_{a}=0.01$, $\delta_u=0.03$, $\delta=1$, $\epsilon=10^{-5}$ and $\eta = 0.4$. For those parameters Lemma~\ref{equilibrium} ensures the existence of one equilibrium, $\textbf{e}^*(m)=(a^*(m),u^*(m),i^*(m))$, which in this case is unique and lies inside the plane \eqref{plane}. For the computation of the Poincaré map we consider the plane $\Sigma:=\{g(a,u,i)=a-a^*(m)=0\}$ as a Poincaré section near the limit cycles.

\begin{figure}[t]
\begin{subfigure}{.33\textwidth}
  \centering
  \includegraphics[scale=.4]{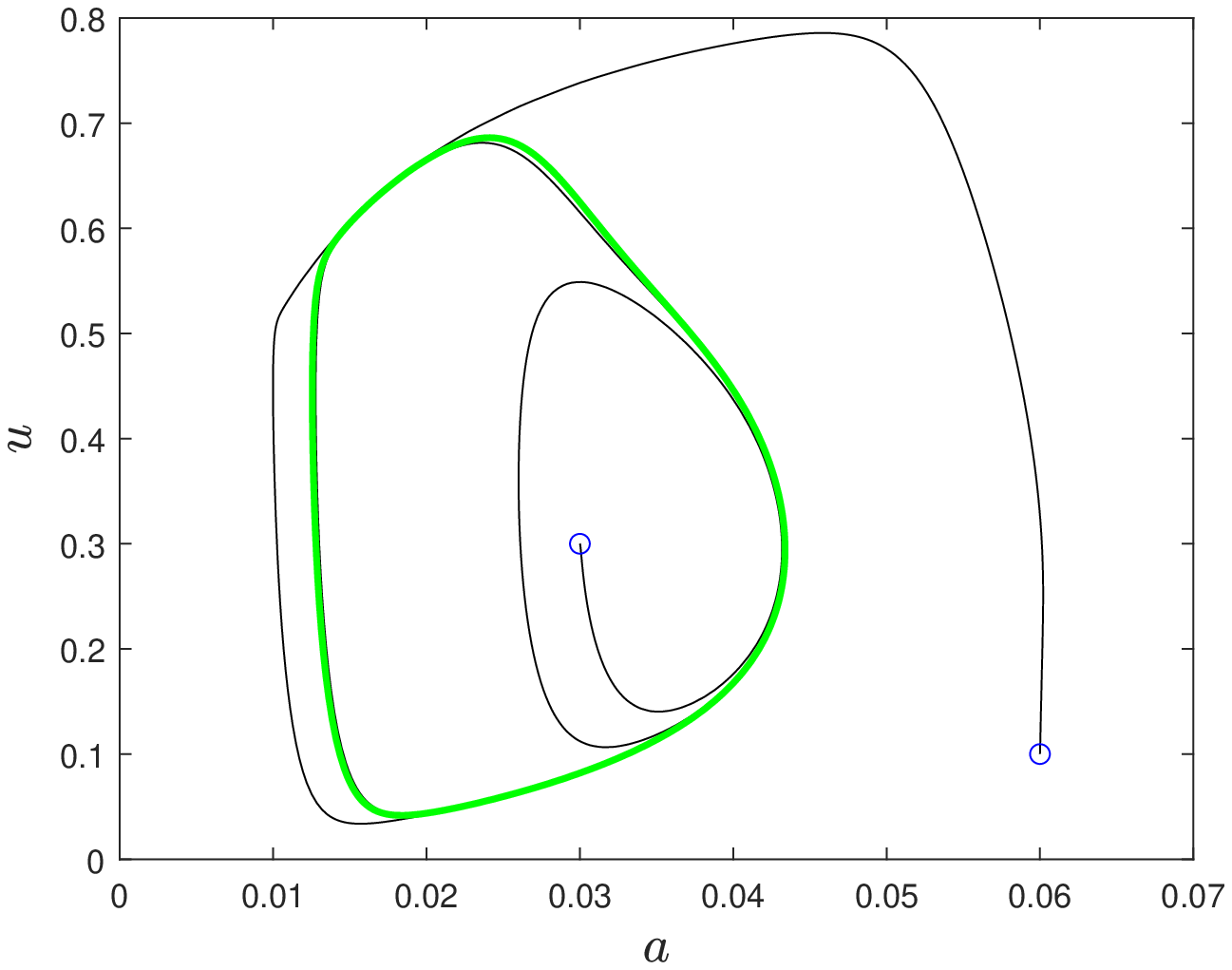}  
\end{subfigure}
\begin{subfigure}{.33\textwidth}
  \centering
  \includegraphics[scale=.4]{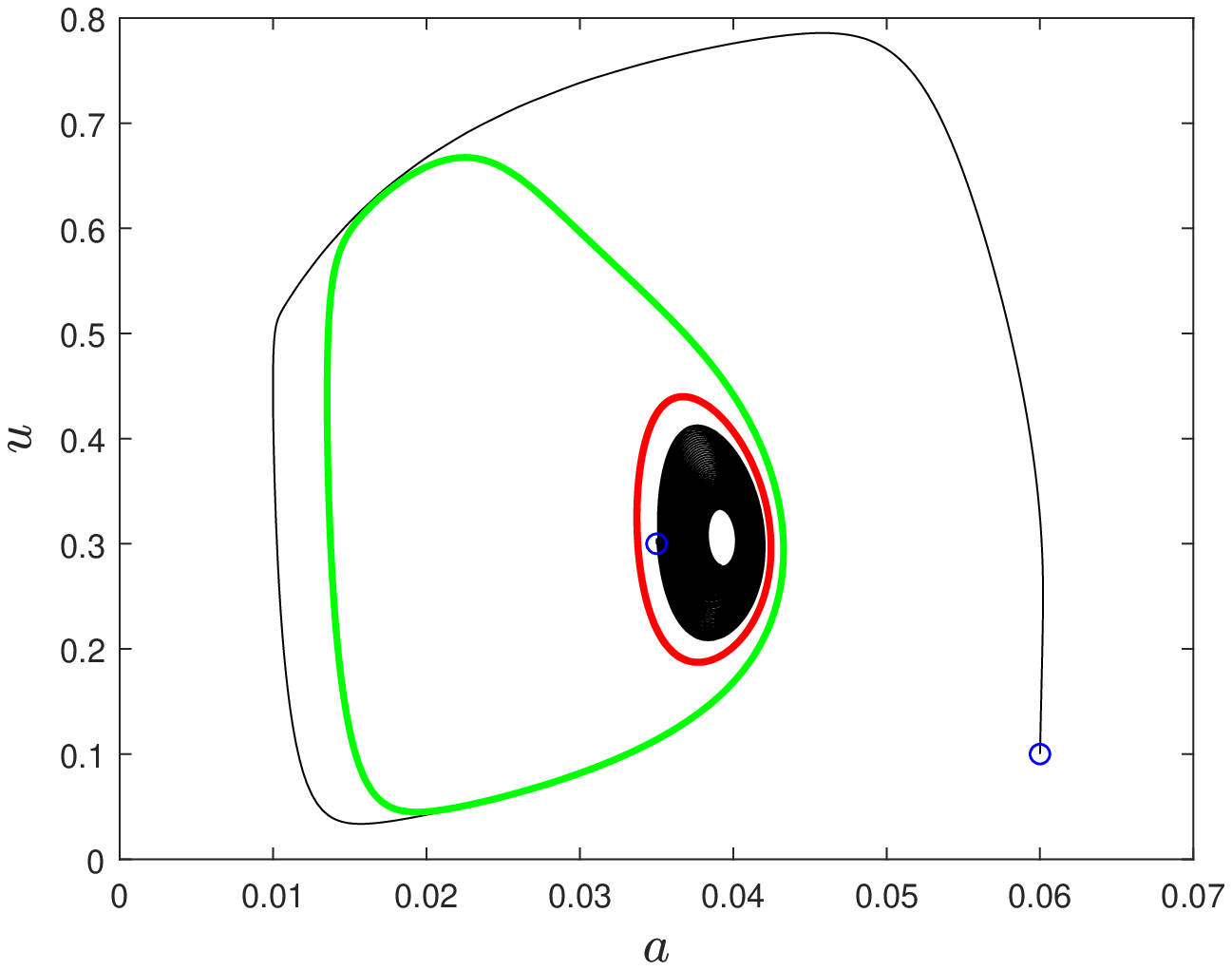}  
\end{subfigure}
\begin{subfigure}{.33\textwidth}
  \centering
  \includegraphics[scale=.4]{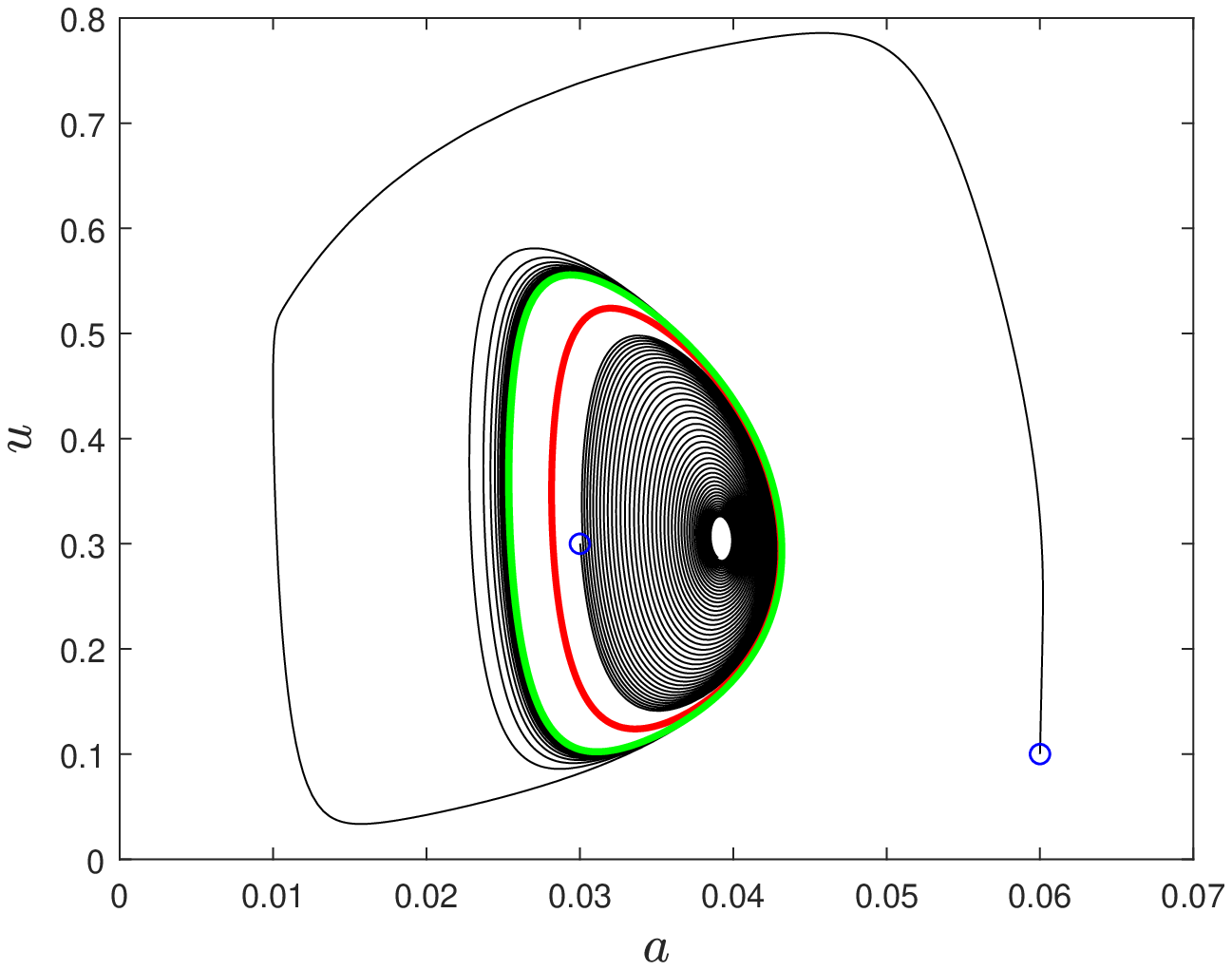}  
\end{subfigure}\\
\begin{subfigure}{.33\textwidth}
  \centering
  \includegraphics[scale=.4]{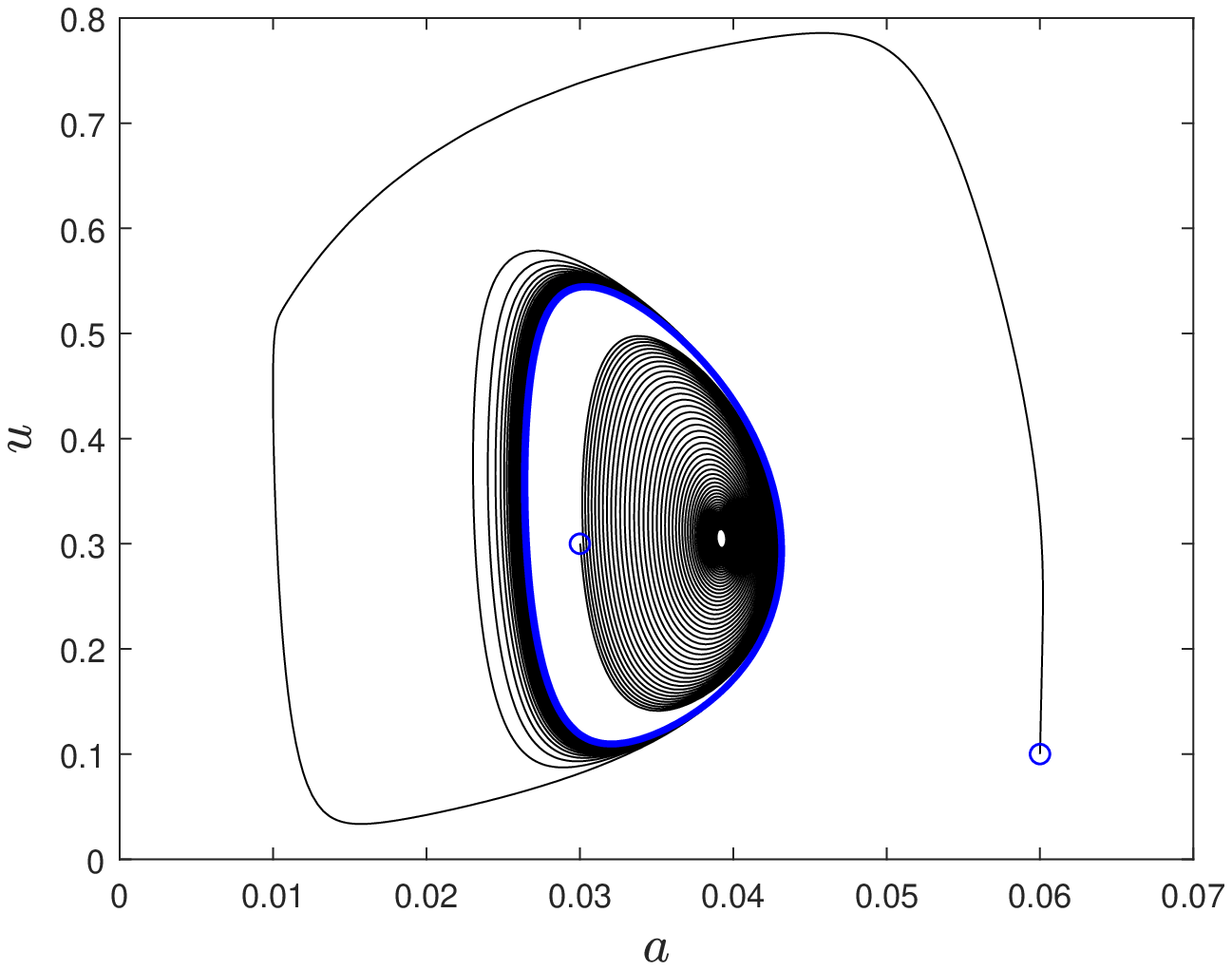}  
\end{subfigure}
\begin{subfigure}{.33\textwidth}
  \centering
  \includegraphics[scale=.4]{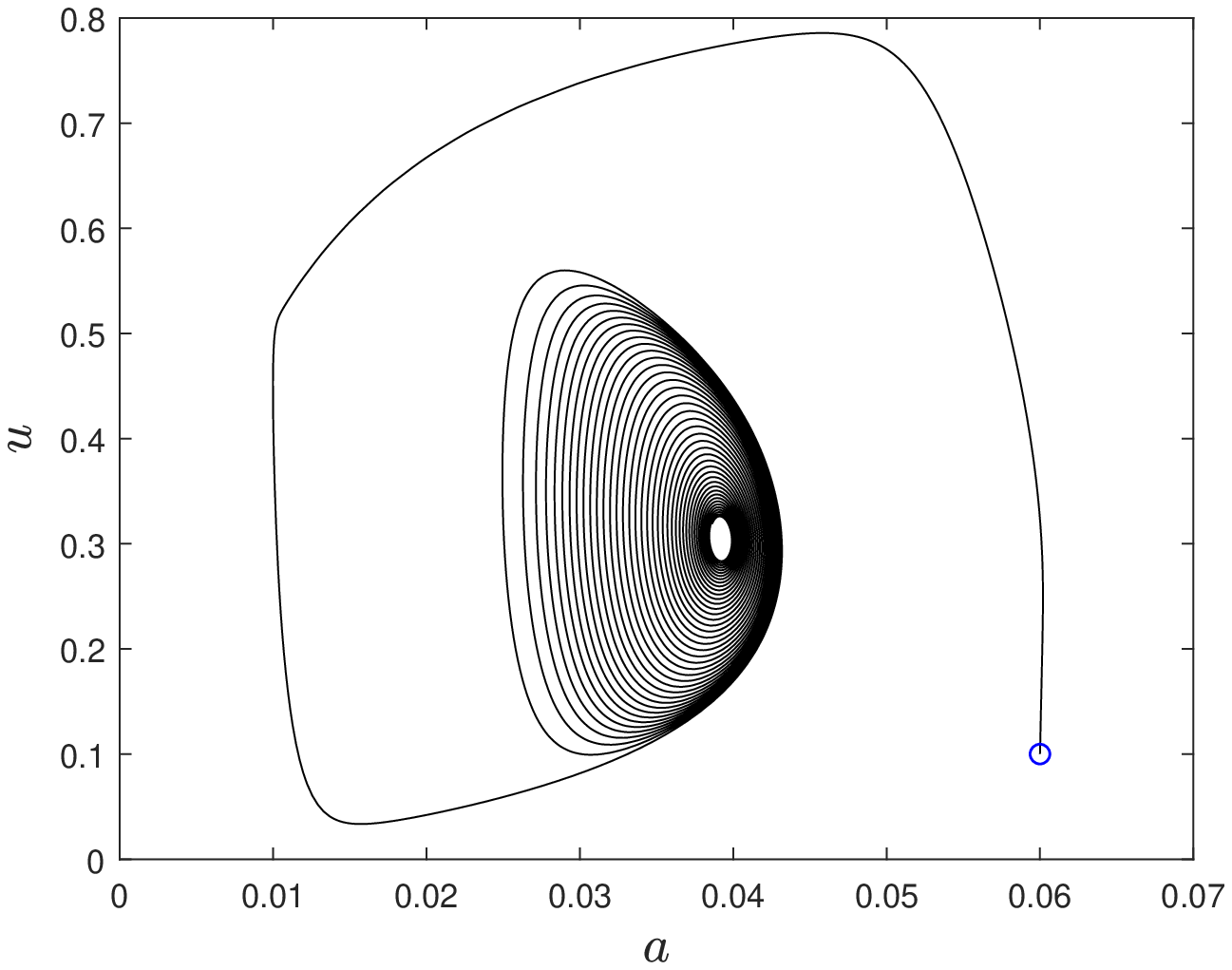}
\end{subfigure} 
\begin{subfigure}{.33\textwidth}
  \centering
  \includegraphics[scale=.4]{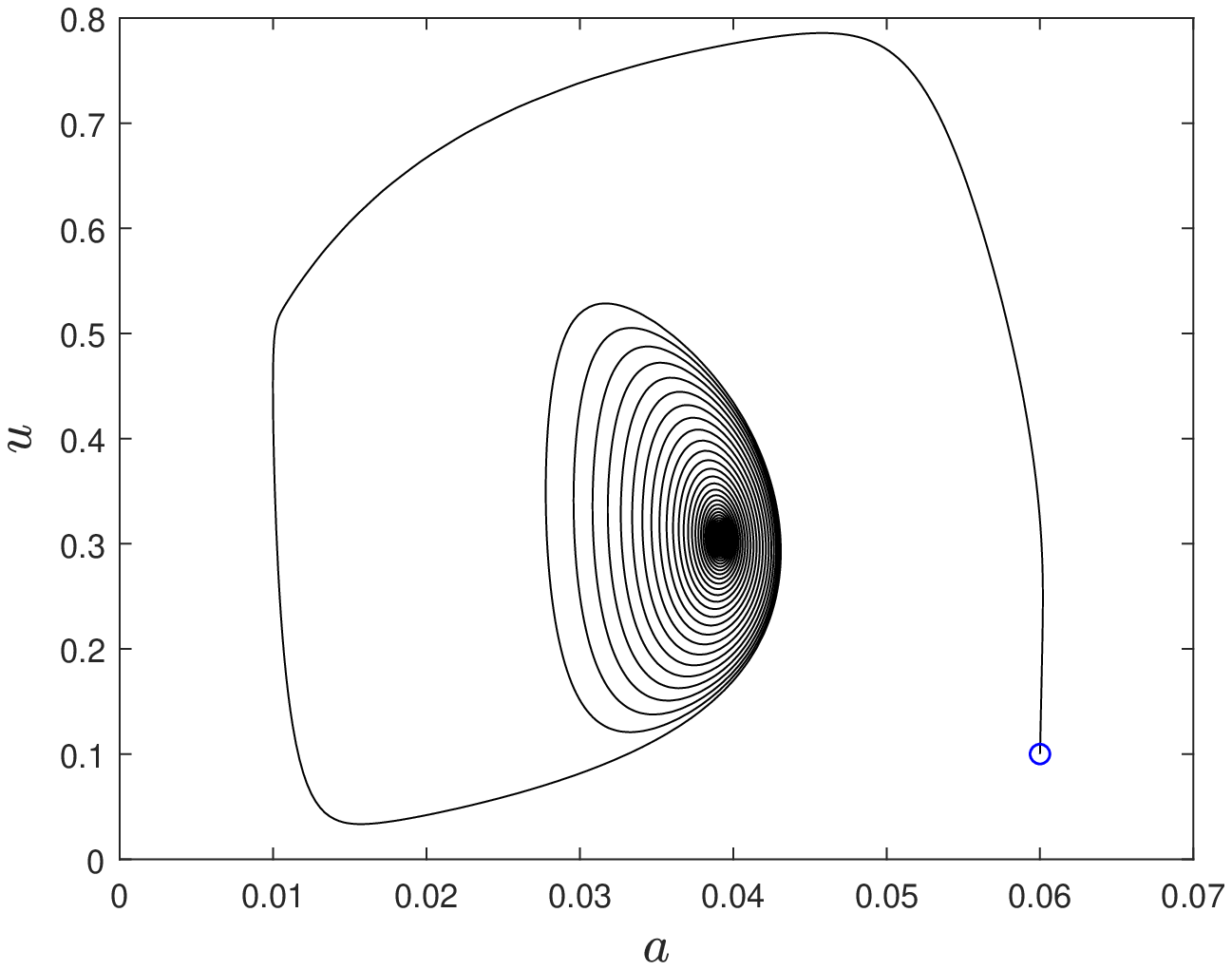}
\end{subfigure} 
\caption{Projections on the $(a,u)$-plane of some solutions of system~\eqref{eqn:SAUIS} for $m=3.69$ (top left), $m=3.73$ (top middle), $m=3.76$ (top right), $m=3.76082$ (bottom left), $m=3.77$ (bottom middle) and $m=3.8$ (bottom right). The orbit in green corresponds to the hyperbolic attractive limit cycle, the orbit in red to the hyperbolic repulsive limit cycle and the orbit in blue to the non-hyperbolic semistable limit cycle. Blue circles mark the initial conditions.}
\label{fig:retrats_fase}
\end{figure}

For values of the parameter $m<m_h\approx 3.698$, the equilibrium $\mathbf{e}^*(m)$ is unstable and the system presents an attractive limit cycle that is mostly restricted to that plane (see Figure~\ref{fig:orbit_plane}). As the parameter increases, the stability of the equilibrium changes, producing a subcritical Hopf bifurcation at $m=m_h$ and a bistability scenario for $m>m_h$: the equilibrium and the limit cycle (see top panels in Figure~\ref{fig:retrats_fase}). As a consequence of the subcritical character of the Hopf bifurcation, an unstable limit cycle is born from the equilibrium. As long as the parameter continues increasing, both limit cycles (stable and unstable) start approaching each other. At the saddle-node bifurcation point $m=m_{sn1}\approx 3.761$, both limit cycles collide in a non-hyperbolic semistable limit cycle and disappear, leaving the equilibrium as the only stable scenario (see bottom panels in Figure~\ref{fig:retrats_fase}). The dynamics remain similar for $m_{sn1}<m<m_{sn2}$ and, at the second bifurcation point $m=m_{sn2}\approx 16.057$, a second saddle-node bifurcation of limit cycles occurs. In this case, a non-hyperbolic semistable limit cycle appears and splits into two hyperbolic (stable and unstable) limit cycles as $m>m_{sn2}$, and this happens without any change of the stability of the equilibrium $\textbf{e}^*(m)$ which remains always asymptotically stable (see Figure~\ref{fig:retrats_fase_m_gran}). The limit cycles separate each other until a position which is qualitatively unchanged as $m$ increases. In Figure~\ref{fig:bifurcation_diagram} we show the previous described bifurcation phenomena, where the amplitudes of the stable limit cycle (solid line) and unstable limit cycle (dashed line) with respect to the proportion of infected nodes are displayed. The amplitude is computed as the difference between the largest and smallest value of the proportion of infected nodes along the orbit. Zero amplitude corresponds to the equilibrium $\textbf{e}^*(m)$. In Figure~\ref{fig:lyapunov_exp} we represent the Floquet characteristic multipliers of the monodromy matrix $D\varphi_{T(x)}(\bm{x})$ of the stable (solid blue line) and unstable (dashed red line) orbits, showing the stability of each limit cycle. We point out that, since $D\varphi_{T(x)}(\bm{x})$ is the monodromy matrix of a limit cycle, one of its eigenvalues is always 1 (the one with eigenvector orthogonal to the section $\Sigma$). Moreover, since the motion is rapidly almost captured by the plane \eqref{plane}, a second eigenvalue is close to zero. The stability of the limit cycles is then given by the remaining third eigenvalue. On the left-hand panel we can see how the unstable limit cycle appears for $m=m_{h}$, with its third eigenvalue being larger than one. When approaching $m=m_{sn1}$ the unstable eigenvalue tends to 1, as it does the third eigenvalue of the stable limit cycle, producing the saddle-node bifurcation of limit cycles and the semi-stability of the orbit. On the right-hand panel we can see that at $m=m_{sn2}$ the semi-stable limit cycle appears giving birth to the stable and unstable limit cycles for $m>m_{sn2}$. We point out the strenght of the unstability in this case, as we also see on the bottom right panel of Figure~\ref{fig:retrats_fase_m_gran}, where orbits are rapidly pushed away from the unstable limit cycle.

\begin{figure}[th]
\begin{subfigure}{.33\textwidth}
  \centering
  \includegraphics[scale=.4]{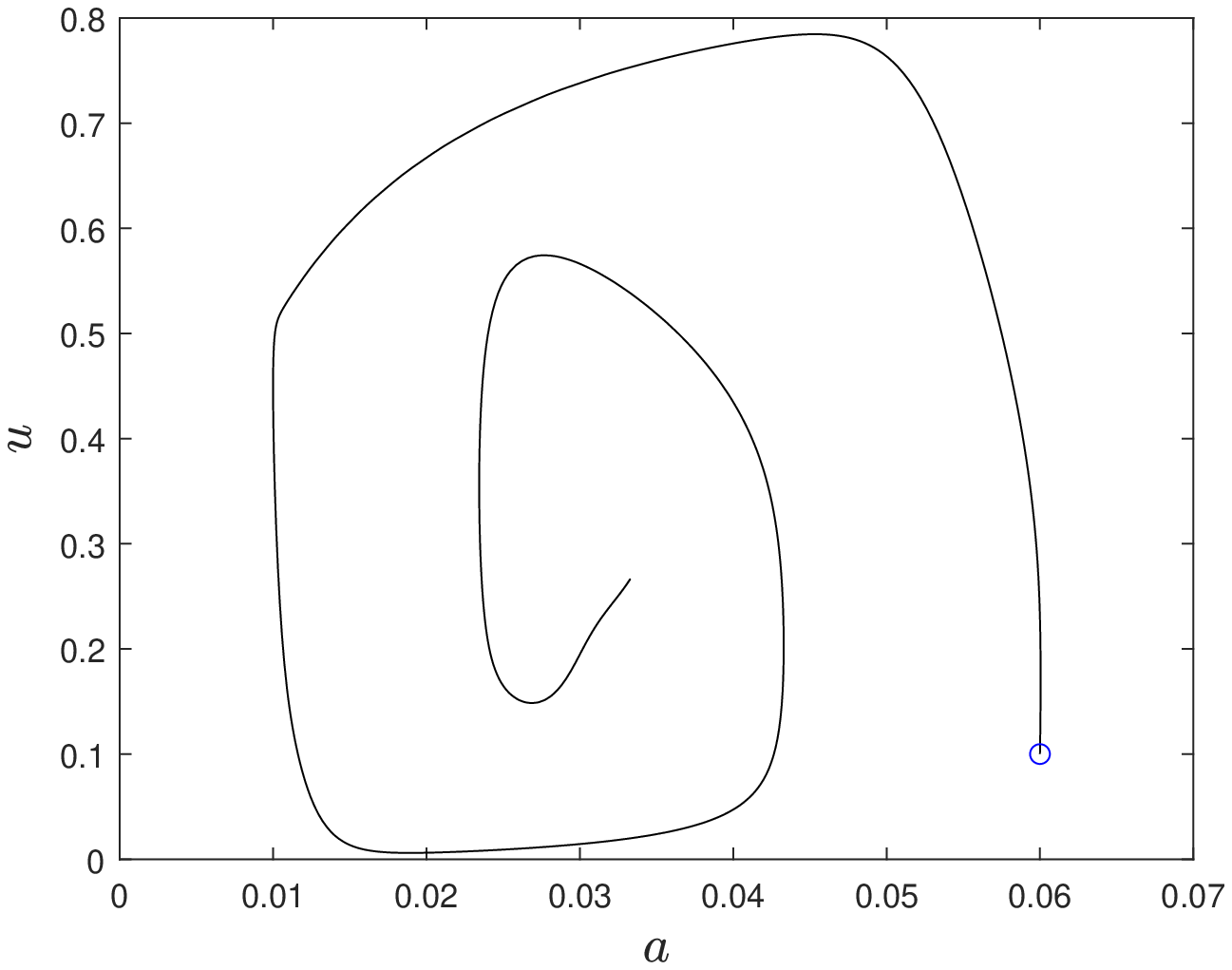}  
\end{subfigure}
\begin{subfigure}{.33\textwidth}
  \centering
  \includegraphics[scale=.4]{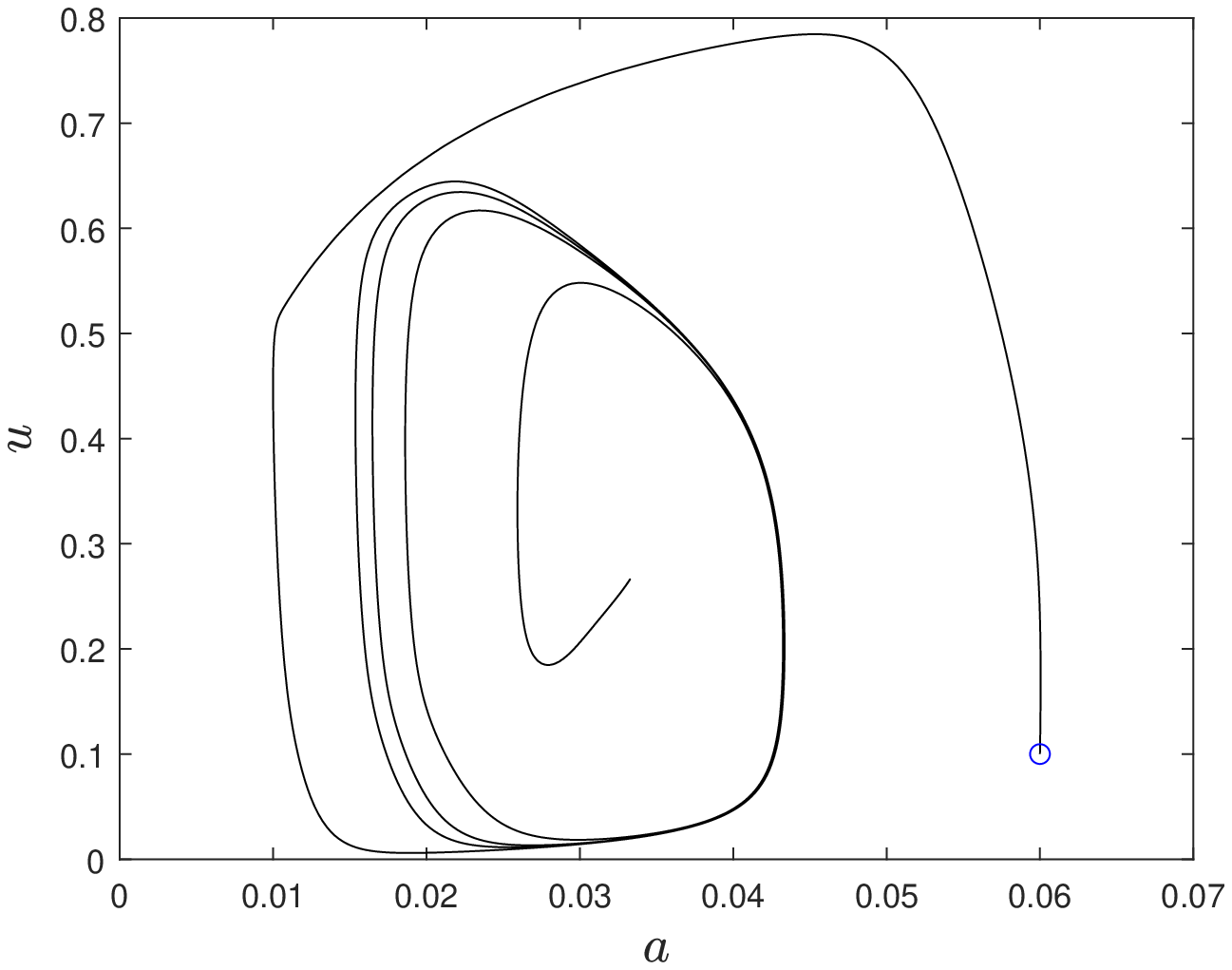}  
\end{subfigure}
\begin{subfigure}{.33\textwidth}
  \centering
  \includegraphics[scale=.4]{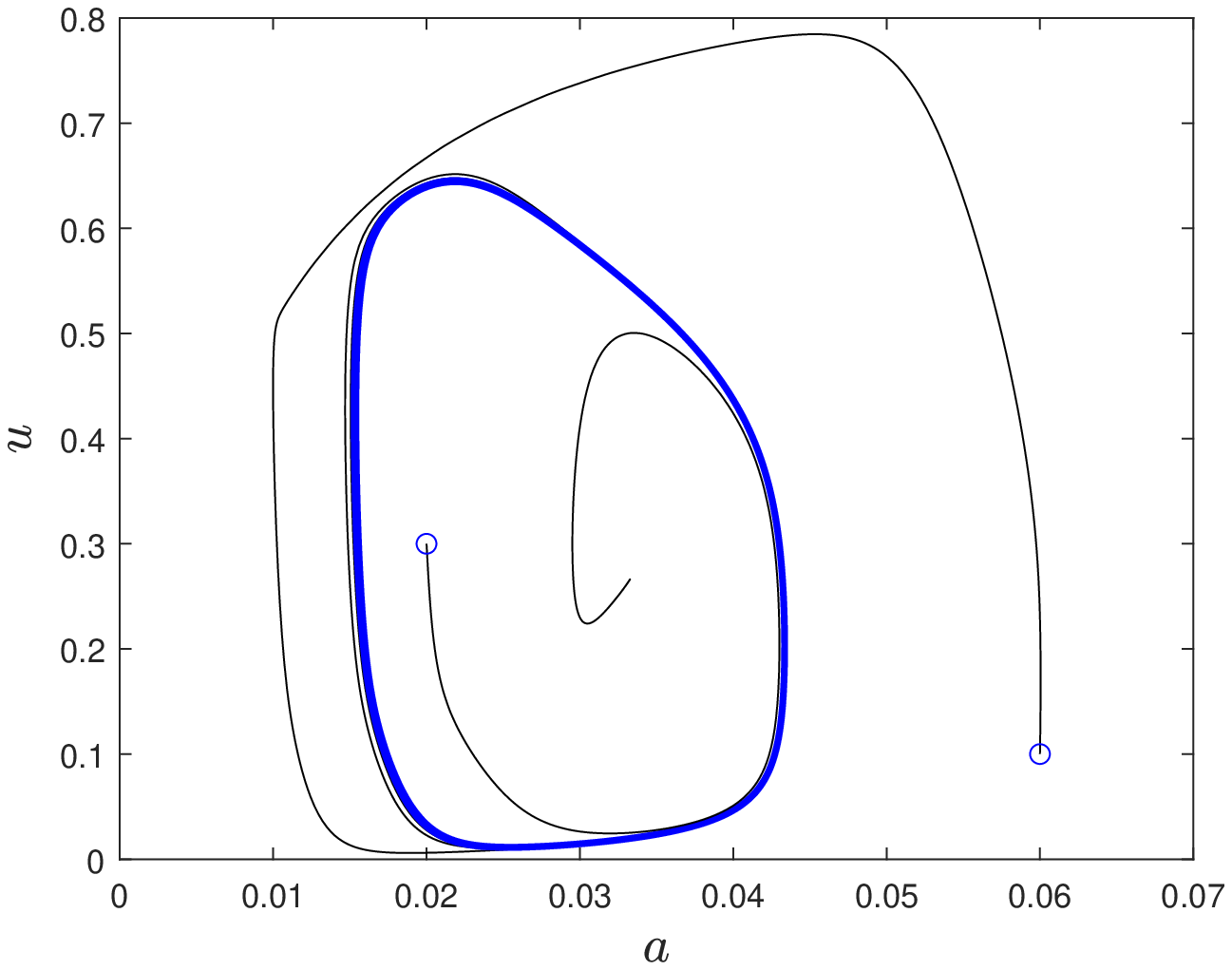}  
\end{subfigure}\\
\begin{subfigure}{.33\textwidth}
  \centering
  \includegraphics[scale=.4]{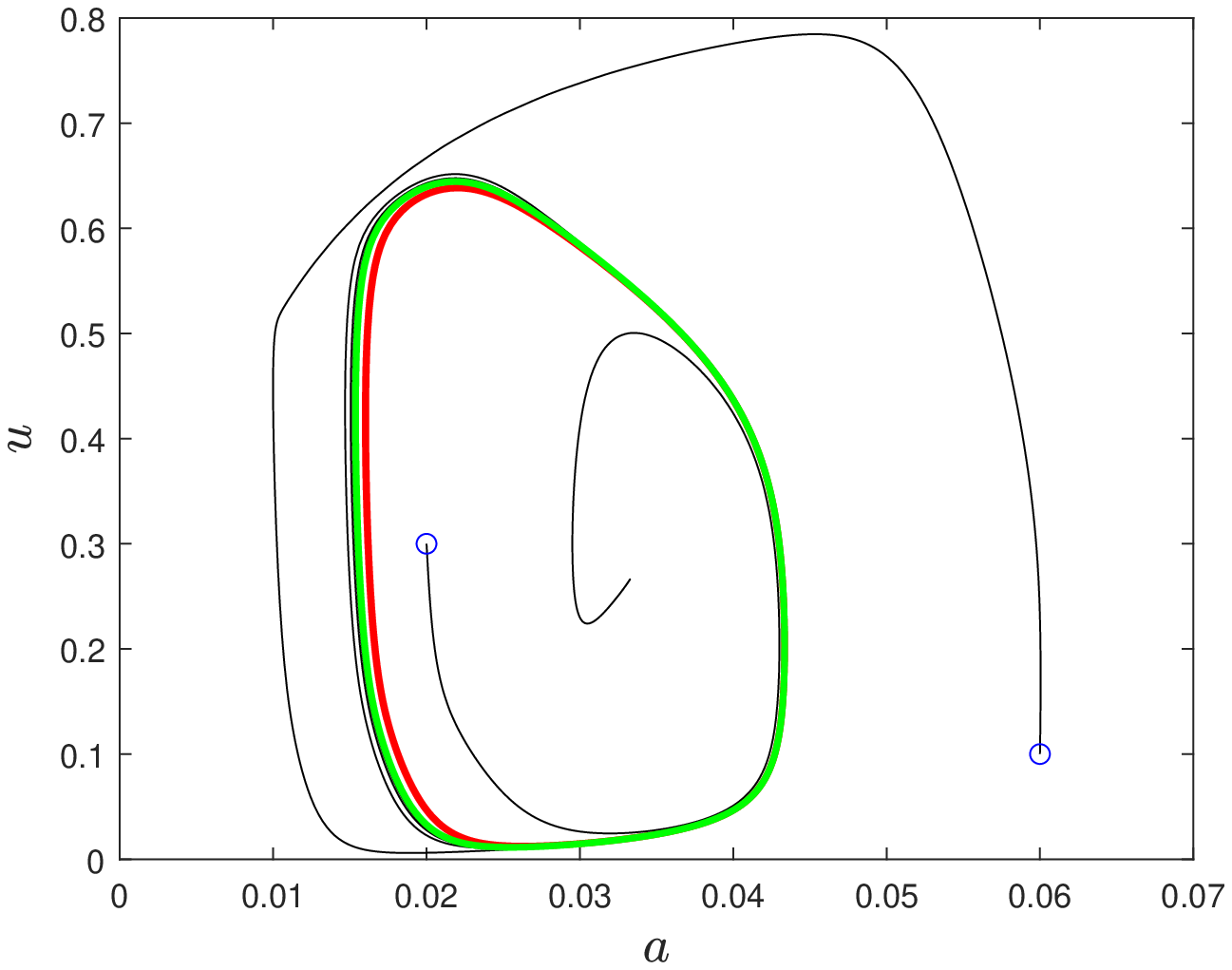}  
\end{subfigure}
\begin{subfigure}{.33\textwidth}
  \centering
  \includegraphics[scale=.4]{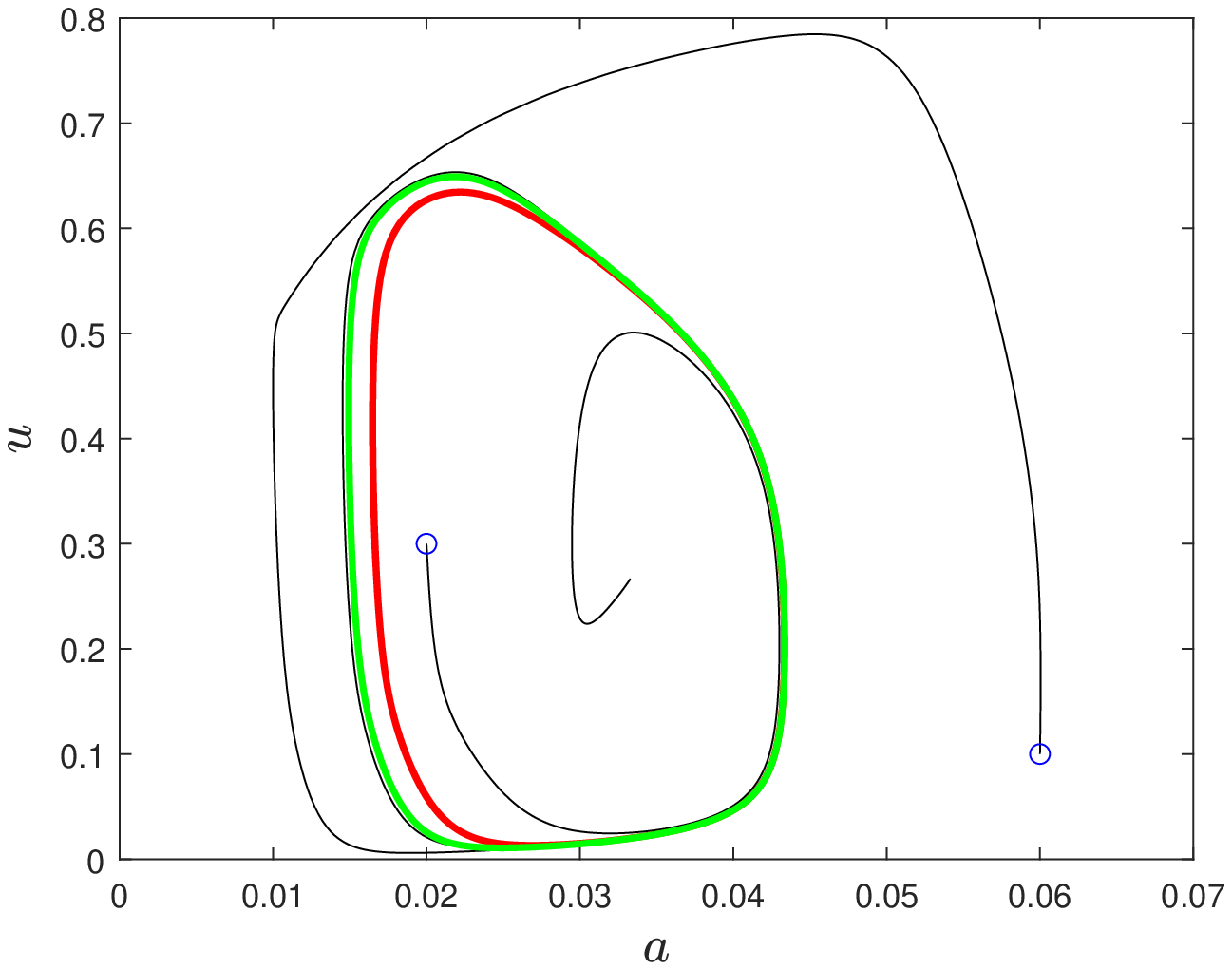}
\end{subfigure} 
\begin{subfigure}{.33\textwidth}
  \centering
  \includegraphics[scale=.4]{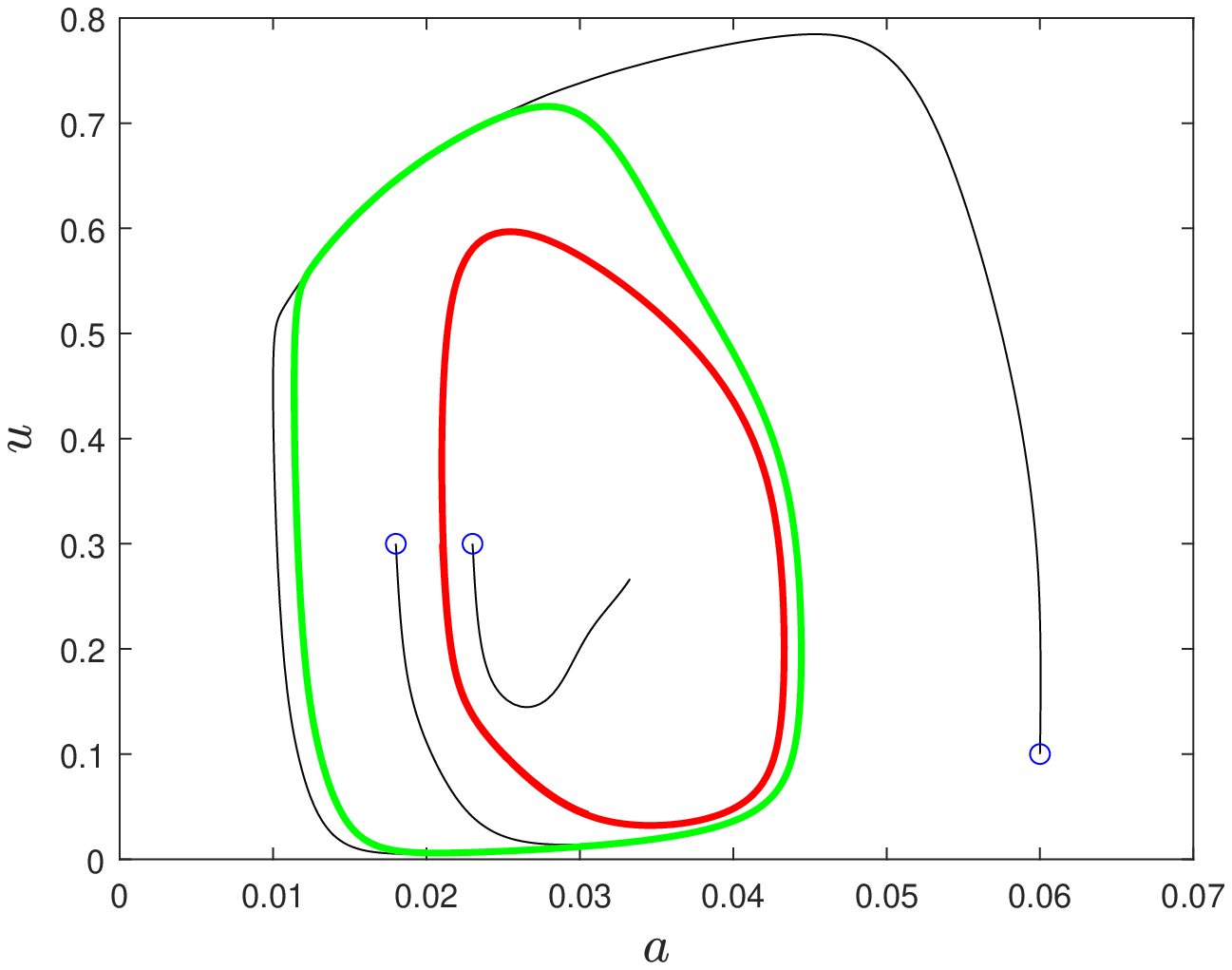}
\end{subfigure} 
\caption{Projections on the $(a,u)$-plane of some solutions of system~\eqref{eqn:SAUIS} for $m=16$ (top left), $m=16.05$ (top middle), $m=16.0575$ (top right), $m=16.05755$ (bottom left), $m=16.06$ (bottom middle) and $m=17$ (bottom right). The orbit in green corresponds to the hyperbolic attractive limit cycle, the orbit in red to the hyperbolic repulsive limit cycle and the orbit in blue to the non-hyperbolic semistable limit cycle. Blue circles mark the initial conditions.}
\label{fig:retrats_fase_m_gran}
\end{figure}

\begin{figure}[ht]
\begin{subfigure}{.49\textwidth}
    \centering
    \includegraphics[scale=0.55]{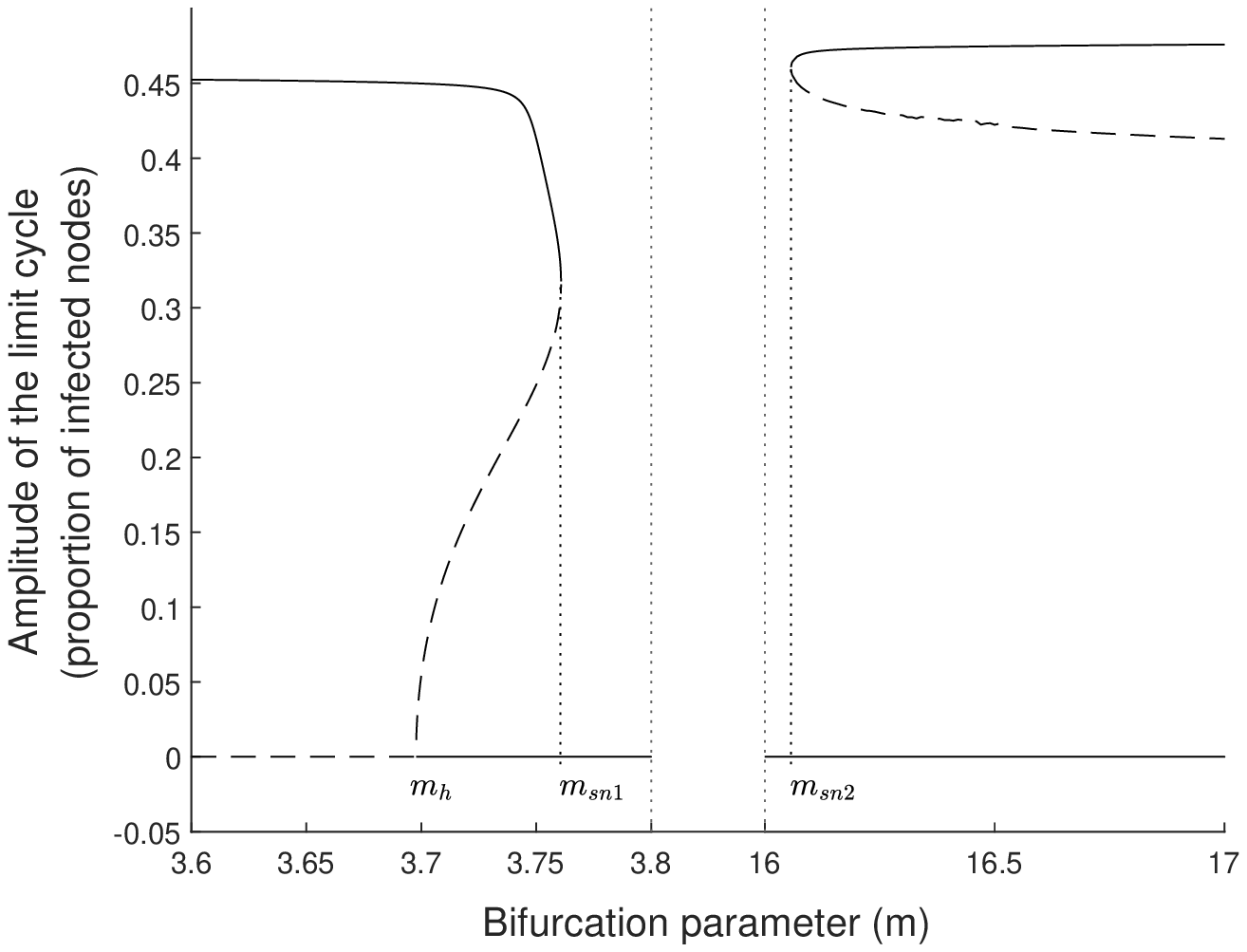}
\end{subfigure}
\begin{subfigure}{.49\textwidth}
    \centering
    \includegraphics[scale=0.55]{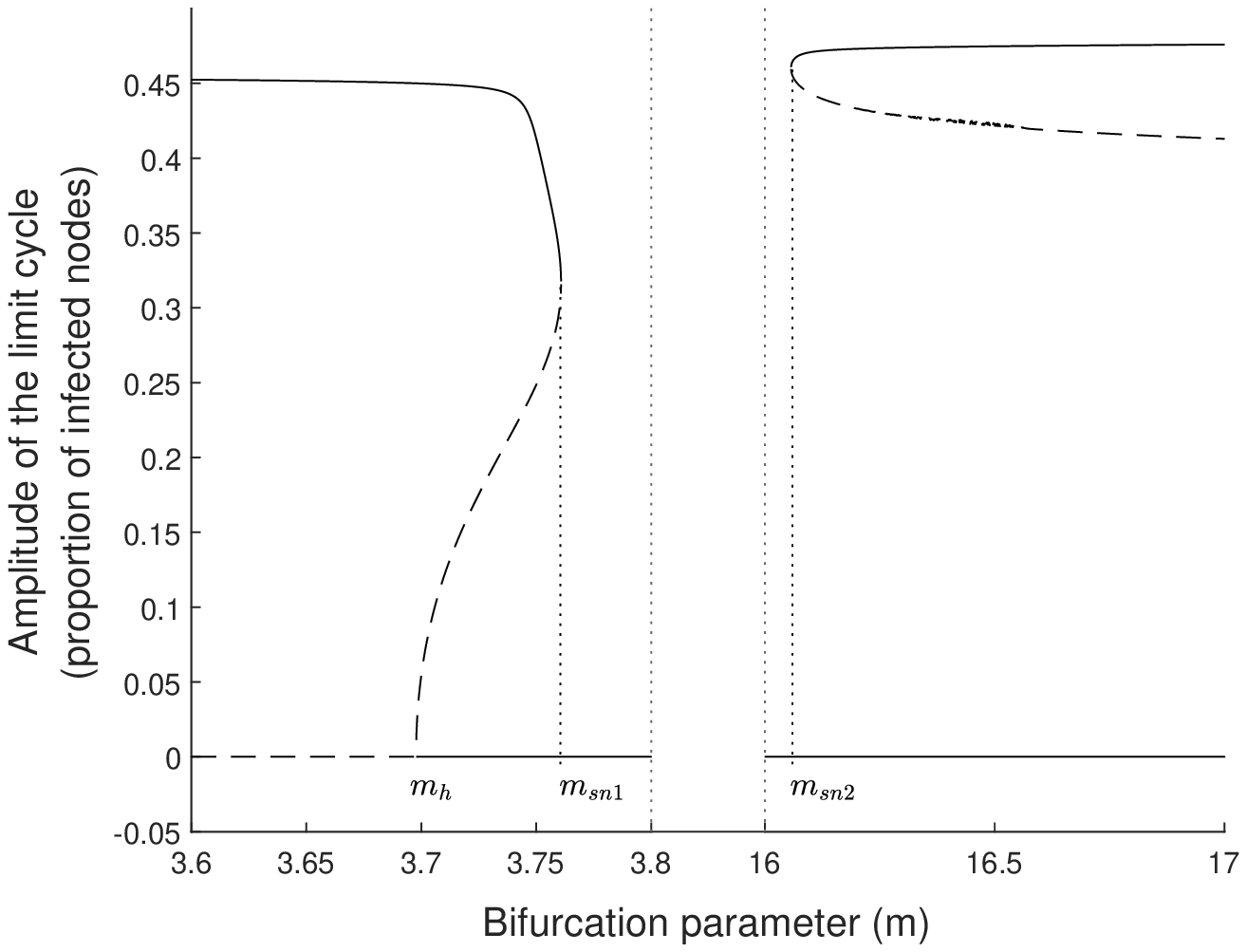}
\end{subfigure}
    \caption{Bifurcation diagram for $\epsilon=0$ (left) and $\epsilon=10^{-5}$ (right).}
    \label{fig:bifurcation_diagram}
\end{figure}

\begin{figure}[ht]
\begin{subfigure}{.49\textwidth}
    \centering
    \includegraphics[scale=0.55]{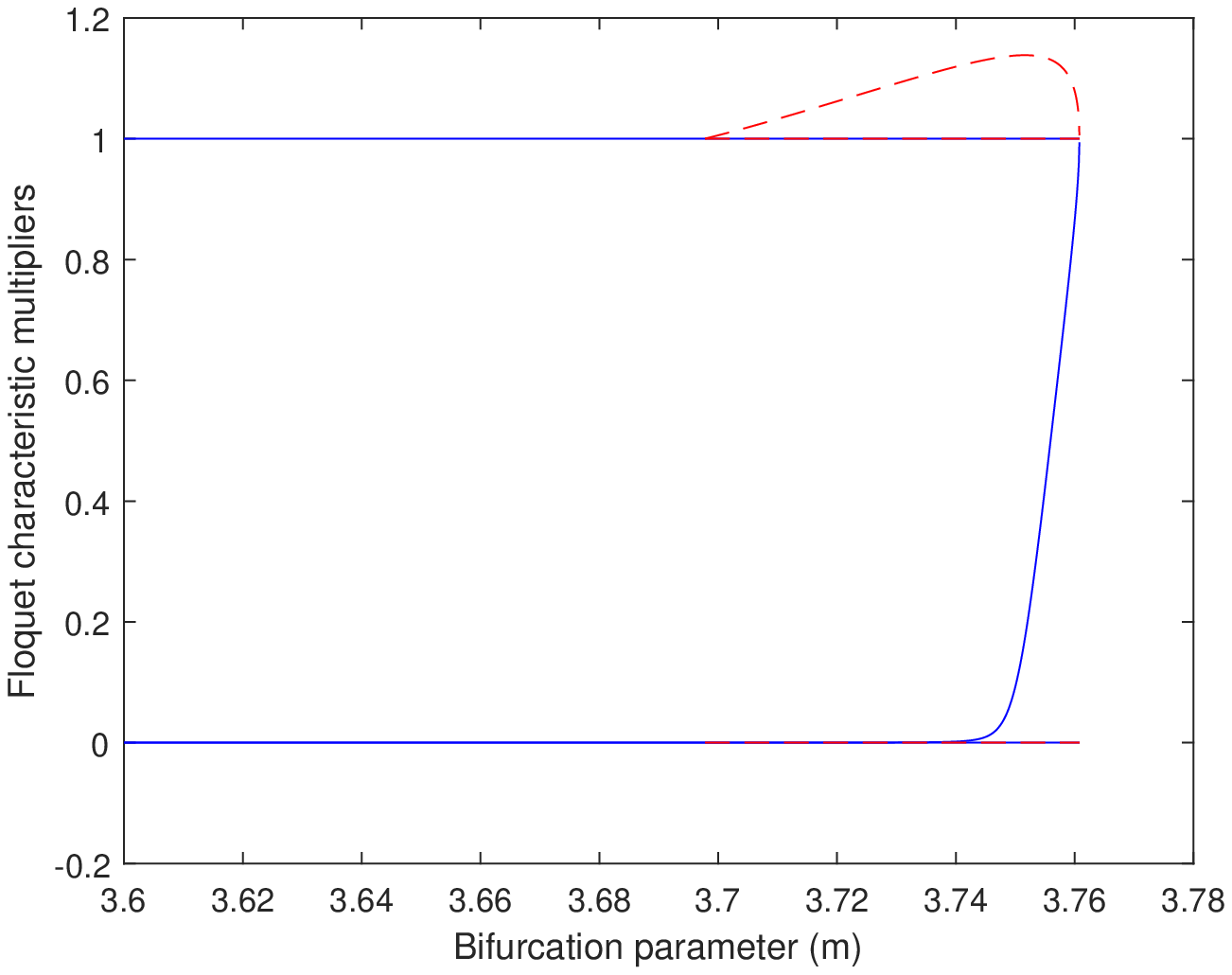}
\end{subfigure}
\begin{subfigure}{.49\textwidth}
    \centering
    \includegraphics[scale=0.55]{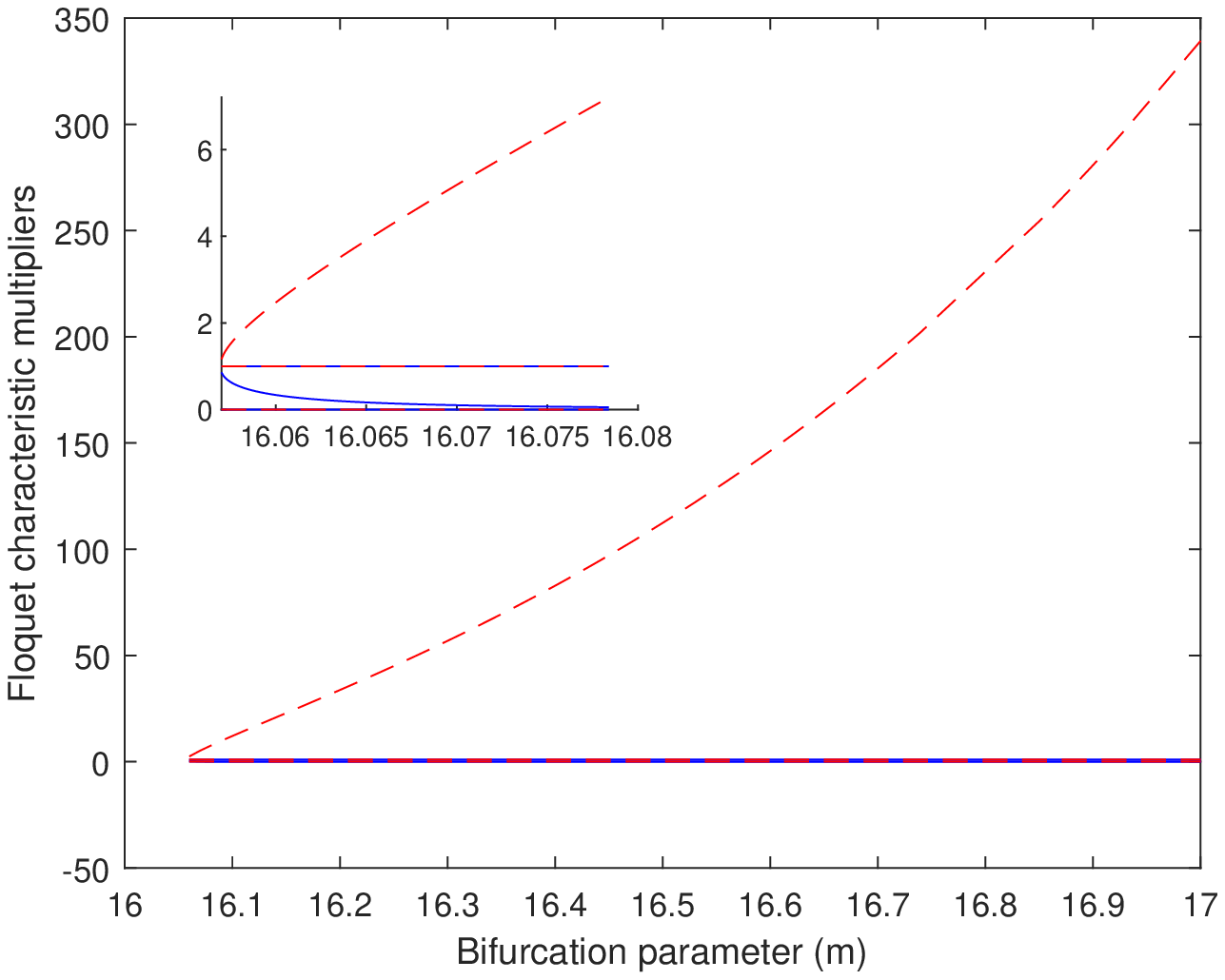}
\end{subfigure}
    \caption{Floquet characteristic multipliers of the stable limit cycle (solid blue line) and the unstable limit cycle (dashed red line).}
    \label{fig:lyapunov_exp}
\end{figure}

\section{Stochastic simulations}
System (\ref{eqn:SAUIS}) is a continuous model for the averaged propagation of an SAUIS-$\epsilon$ epidemics among the individuals in a fully mixed population. So, we need to perform stochastic simulations in order to assess in which sense and to which extent the invariant objects exhibited by system (\ref{eqn:SAUIS}) are found also in a discrete context. 

\subsection{General simulation setup}\label{sec:Simsetup}
As usual in the setting of continuous-time stochastic simulations, we use the well-known Gillespie algorithm (GA in what follows) \cite{Gillespie}, which was originally designed to simulate a fully mixed chemically reacting system. We have a population of $N$ individuals and the algorithm keeps trace of the total numbers $S$, $A$, $I$, $U$ of susceptible, aware, infected and unwilling individuals. In short: at each step, an event is chosen at random according to its \emph{weight} (for instance, an infection $I+S\rightarrow I+I$ has weight $\beta$) over the sum of the weights of all possible events. Once chosen, the event takes place and $A$, $I$, $U$, $S=N-A-I-U$ are accordingly updated. The continuous time is increased by a random positive number drawn from a certain exponential probability distribution \cite{Gillespie,JRS2}.

Given a population size $N$, a combination of model parameters, and an initial condition $a(0),u(0),i(0)$, we run 50 independent simulations, each corresponding to a random distribution of $a(0)N$, $u(0)N$, $i(0)N$ and $(1-a(0)-u(0)-i(0))N$ nodes having respectively the initial states of aware, unwilling, infected and susceptible. For any experiment, we store the evolution of $a(t)$, $u(t)$ and $i(t)$ as three time series of equally-spaced points in the interval $[0,T]$, where $T$ is the maximum running continuous-time of the simulation. All along the paper, the caption of each reported figure obtained by simulation includes the specification of the values of $N$, $a(0)$, $u(0)$, $i(0)$ and $T$.

\subsection{Detection of the bistabiliy regime}\label{sec:Bistability}
With the aim of producing an analogous stochastic version of the numerical bifurcation diagram depicted in Figure~\ref{fig:bifurcation_diagram} we follow the procedure mentioned in the previous section. We perform $50$ independent simulations for each value of the parameter $m$ and initial conditions starting from the location of the equilibrium $\textbf{e}^*(m)=(a^*(m),i^*(m),u^*(m))$, keeping $a(0)=a^*(m)$ fixed, increasing the value of $i(0)$ from $i^*(m)$ by $0.01$ and computing $u(0)$ according to the equality~\eqref{plane}. In this way, for each $m$ approximately $20$ different initial conditions are considered inside a straight line lying on the plane~\eqref{plane}. The larger $i(0)$, the farther from the equilibrium the initial condition.

\begin{figure}[t]
\begin{subfigure}{.33\textwidth}
  \centering
  \includegraphics[scale=.4]{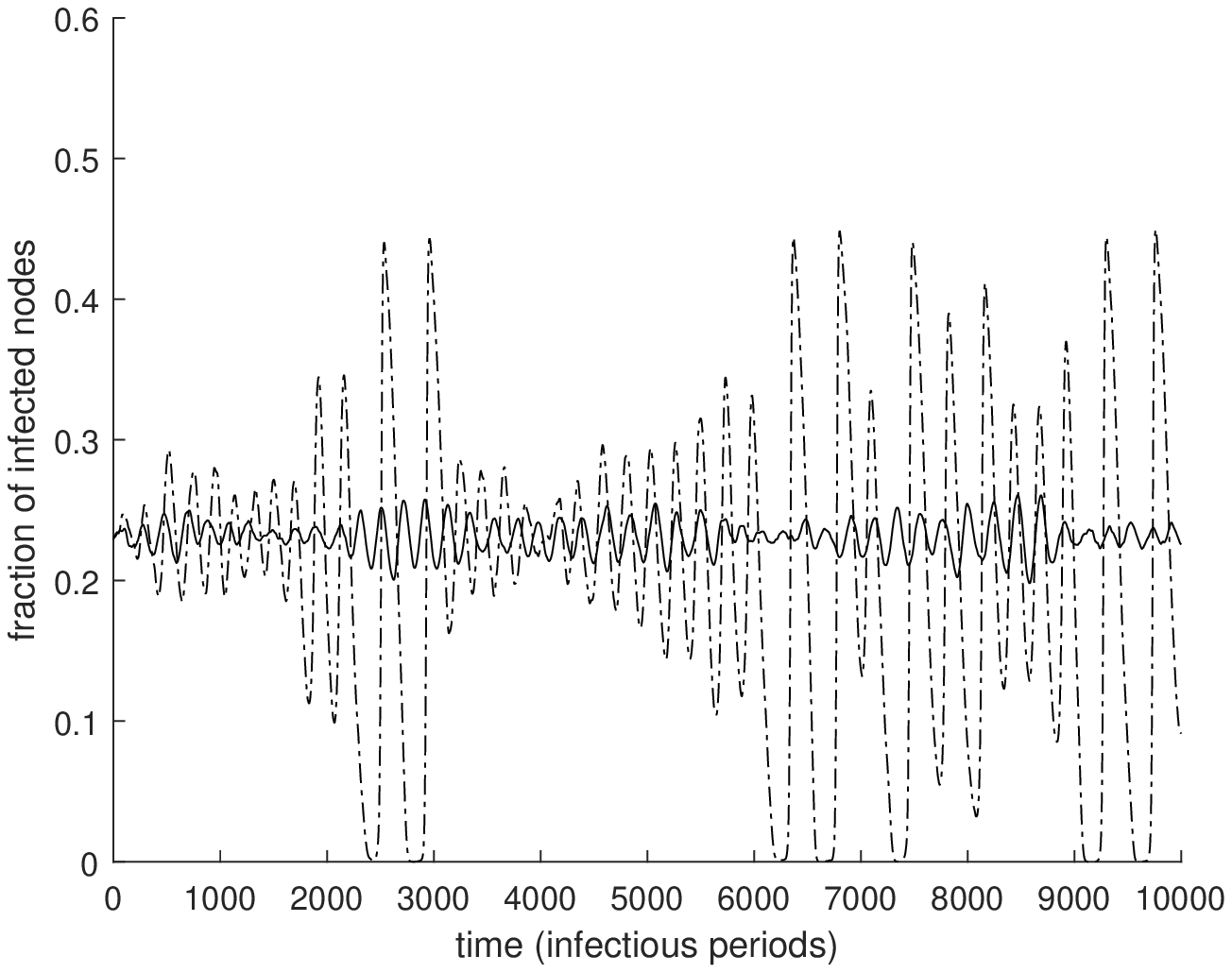}  
\end{subfigure}
\begin{subfigure}{.33\textwidth}
  \centering
  \includegraphics[scale=.4]{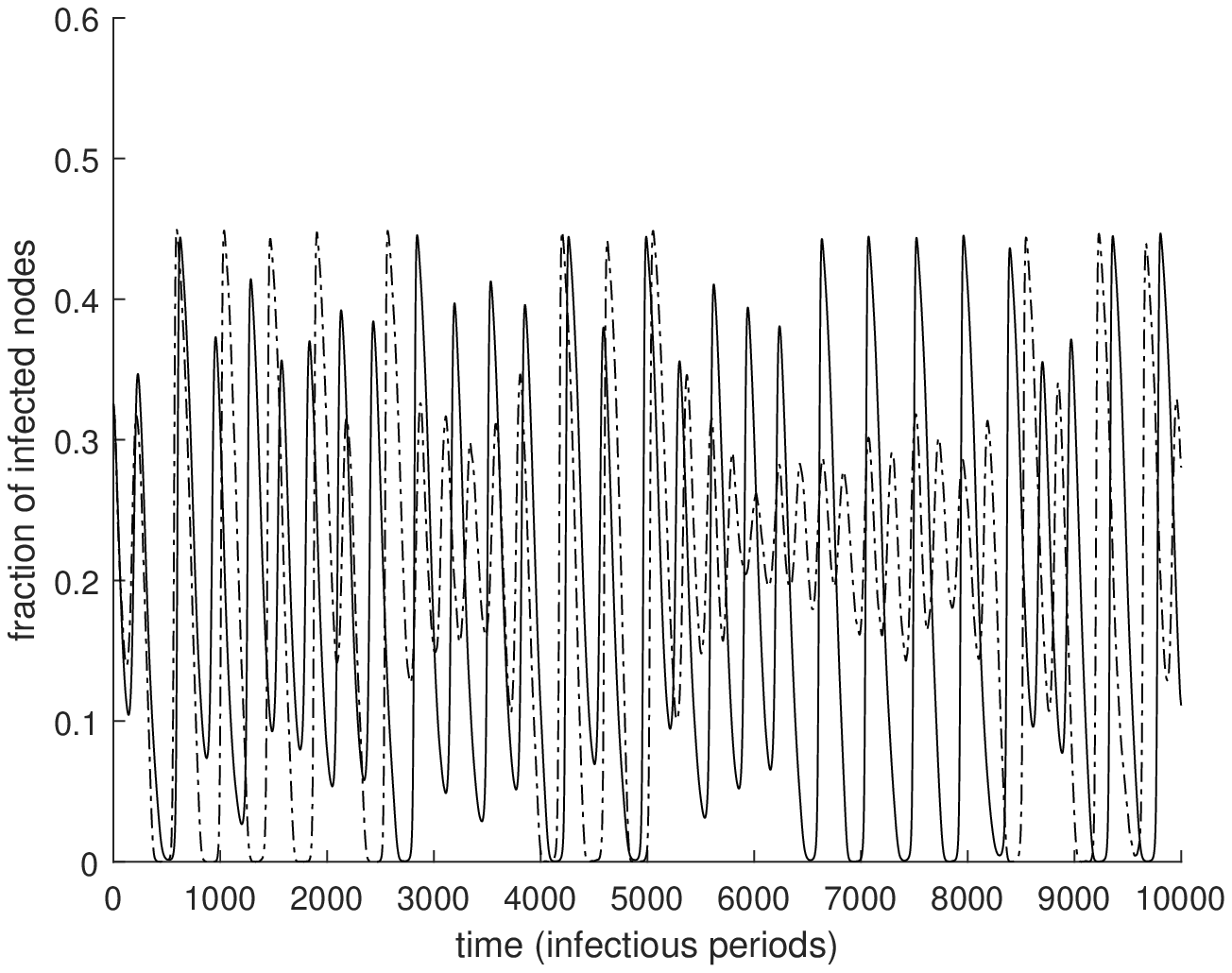}  
\end{subfigure}
\begin{subfigure}{.33\textwidth}
  \centering
  \includegraphics[scale=.4]{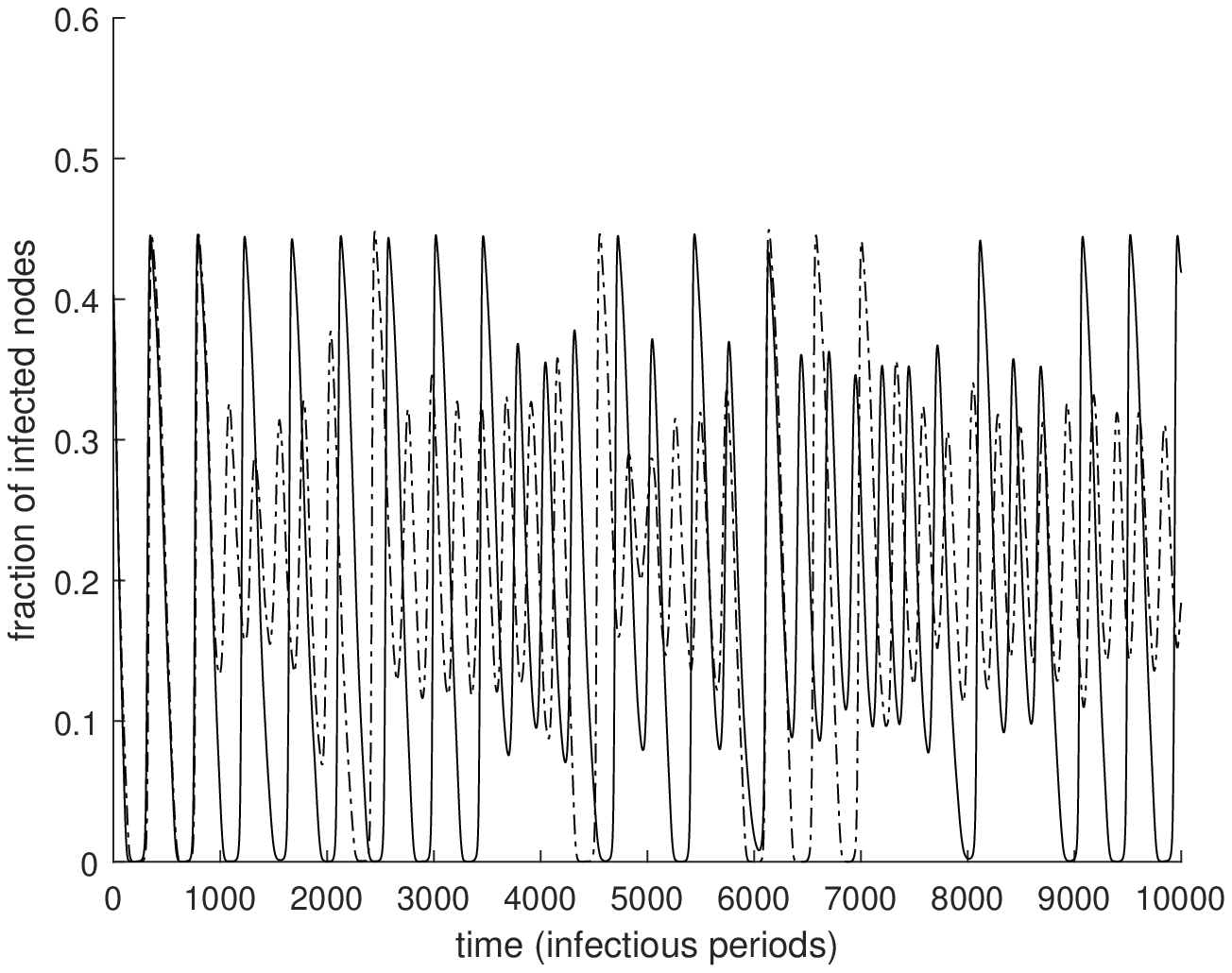}  
\end{subfigure}\\
\begin{subfigure}{.33\textwidth}
  \centering
  \includegraphics[scale=.4]{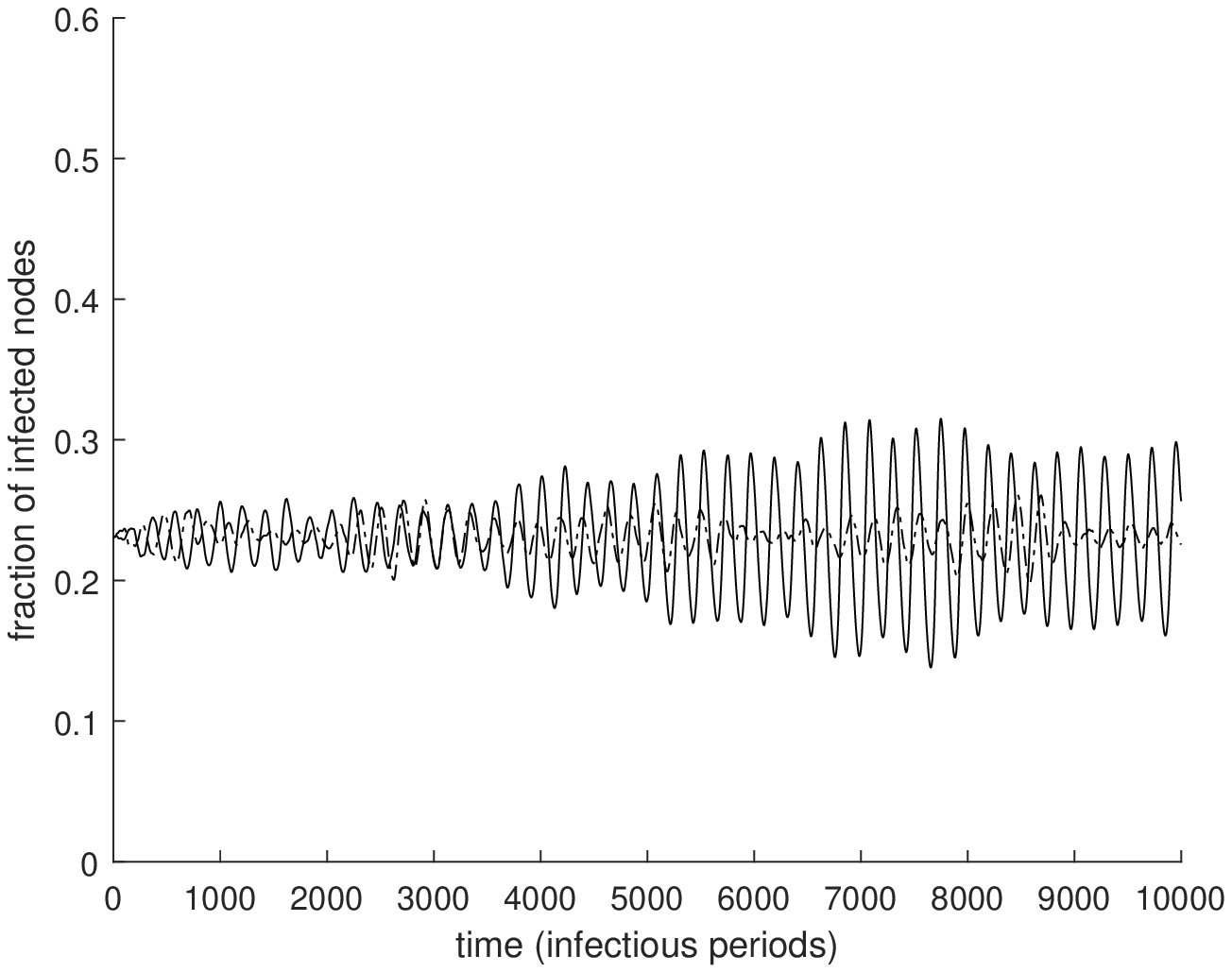}  
\end{subfigure}
\begin{subfigure}{.33\textwidth}
  \centering
  \includegraphics[scale=.4]{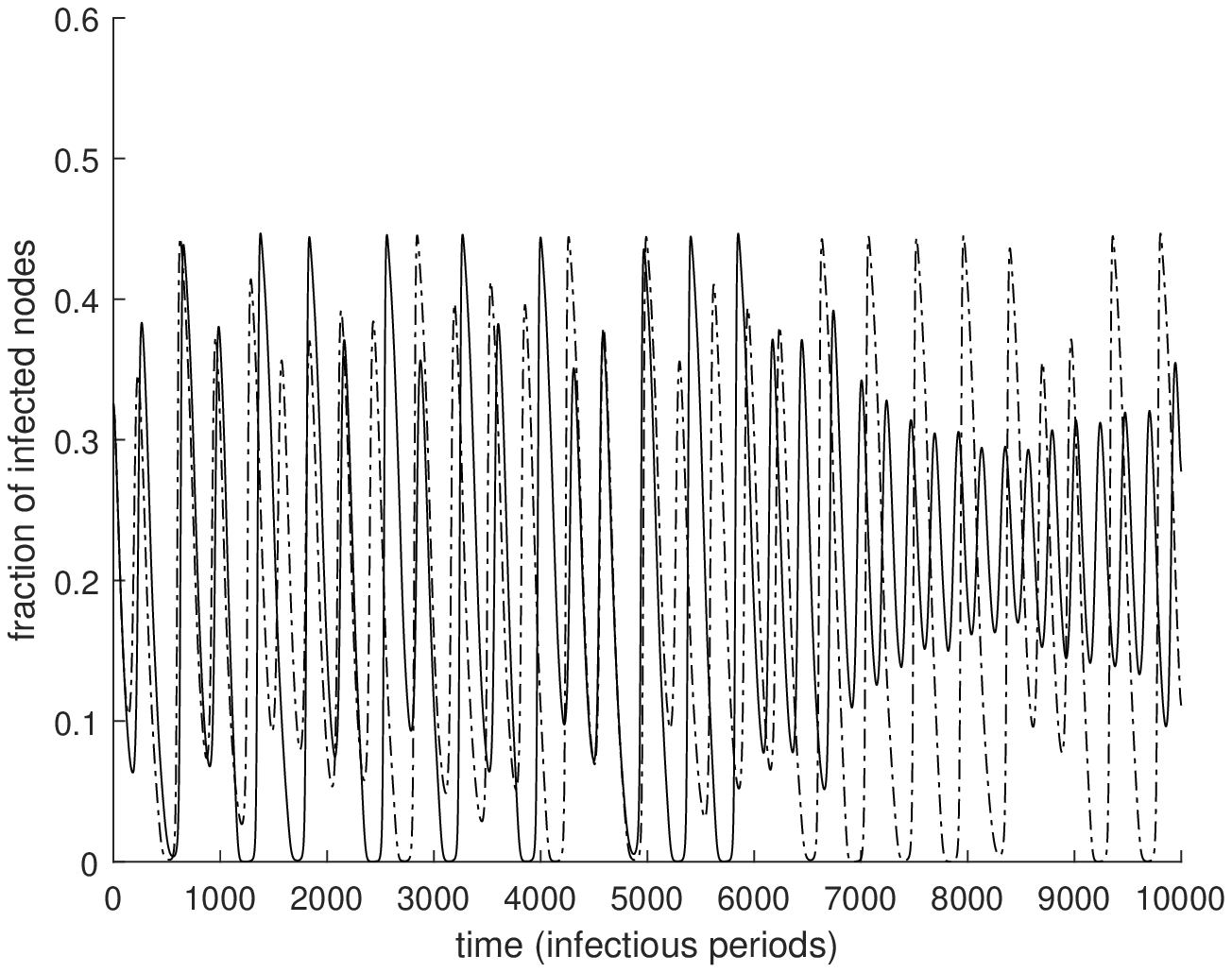}
\end{subfigure} 
\begin{subfigure}{.33\textwidth}
  \centering
  \includegraphics[scale=.4]{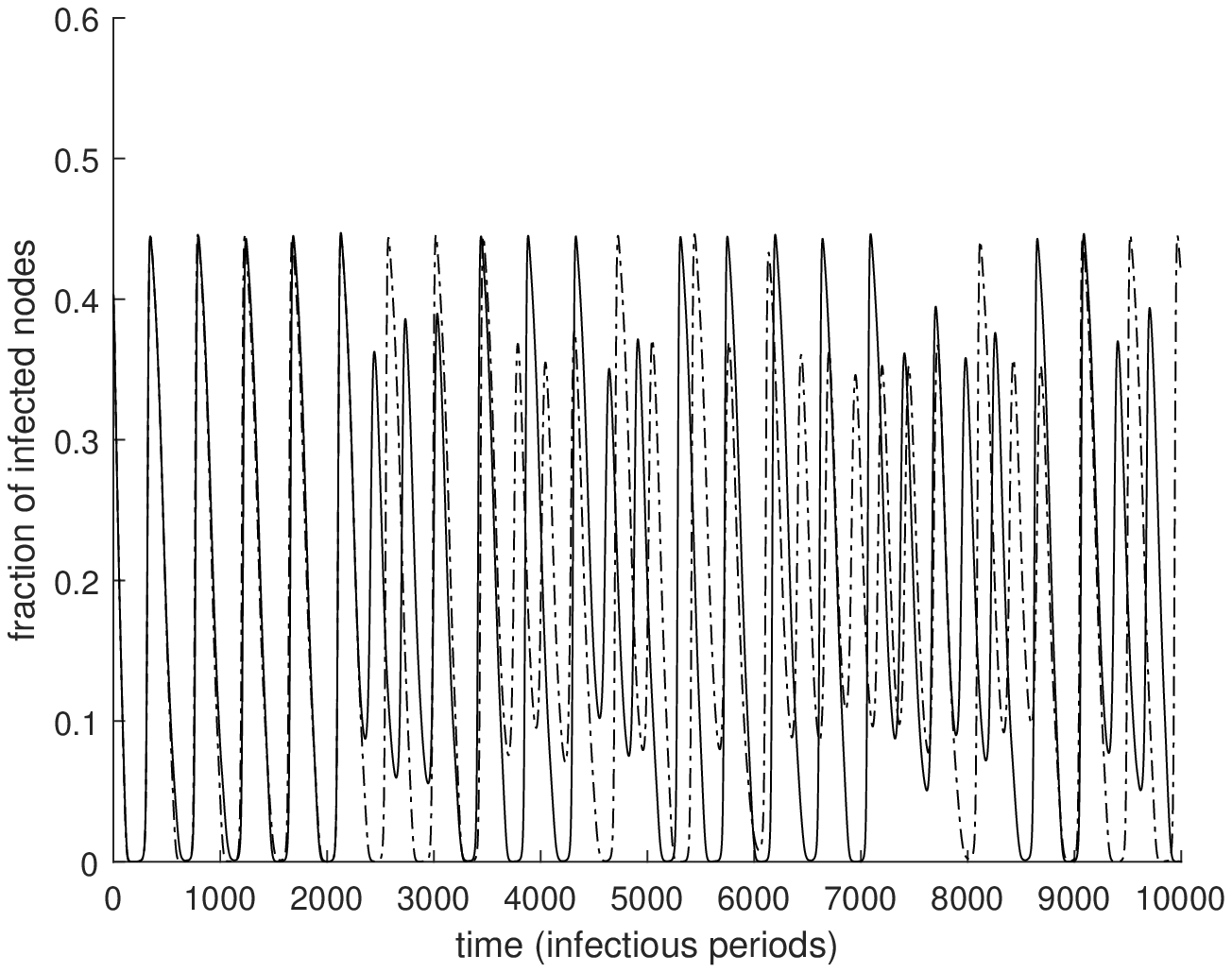}
\end{subfigure} 
\caption{On top, two outputs from GA with parameter $m=3.7$ with $N=10^{5}$ (dotted line) and $N=10^{6}$ (solid line), and initial conditions $a_0=0.039359$, $i_0=0.231104$, $u_0=0.306049$ (left), $a_0=0.039359$, $i_0=0.331104$, $u_0=0.172716$ (center), $a_0=0.039359$, $i_0=0.431104$, $u_0=0.039382$ (right) and same initial random seed. 
On bottom, two outputs from GA with parameter $m=3.7$ with $N=10^{6}$ and same initial conditions as before with different initial random seed.}
\label{fig:experiments_simus_m_petit}
\end{figure}

In stochastic epidemic models the average of the data is the usual way to construct a single signal to compare with the analytic dynamics. This is so because ODE systems are a good approximation of the mean of realizations of stochastic processes in systems with a large number of components. However, although the standard mean of trajectories is a good option when dealing with high prevalence endemic equilibria, it does not always work when a system exhibits fluctuating dynamics and even less in the presence of bistability. In the three lower panels of Figure~\ref{fig:experiments_simus_m_petit} we see two different realizations of the same experiment with identical initial conditions but different initial random seed. Although they start identically, it is clear that stochasticity reveals differences in the global behaviour beyond the expected small perturbations. This phenomena is due to the bistability of the system, where small stochastic perturbations may change a realization from one attraction basin to the other. Of course, the greater the number of nodes $N$, the more similitude between realizations, as we can see in the three top panels of Figure~\ref{fig:experiments_simus_m_petit}. The bistability region in the previous figure is sensitive and very narrow, as shown in Figure~\ref{fig:bifurcation_diagram}. A more clear example is given in Figure~\ref{fig:experiments_simus_m_gran} for the parameter values $m=17$ and $m=18$.

\begin{figure}[t]
\begin{subfigure}{.33\textwidth}
  \centering
  \includegraphics[scale=.4]{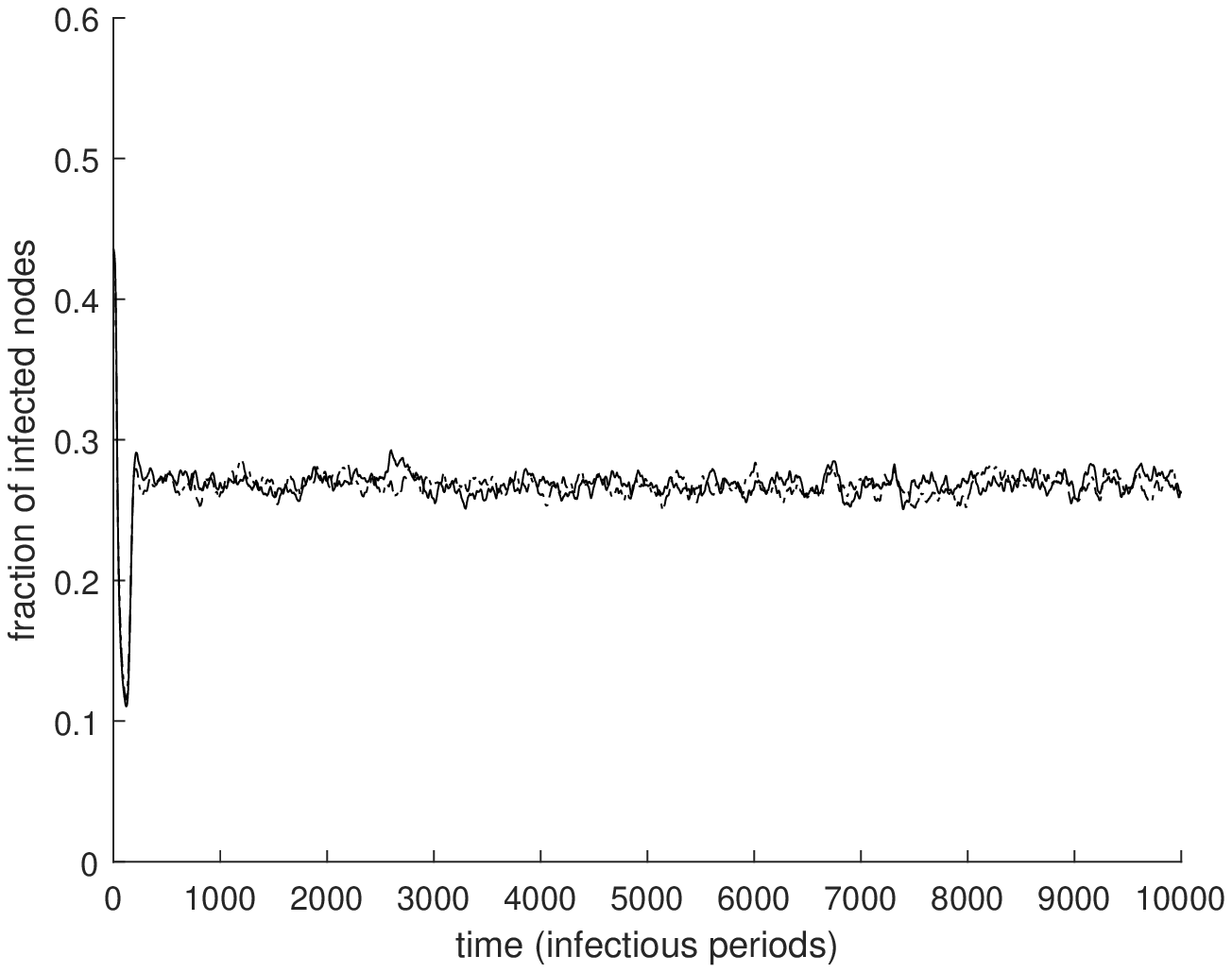}  
\end{subfigure}
\begin{subfigure}{.33\textwidth}
  \centering
  \includegraphics[scale=.4]{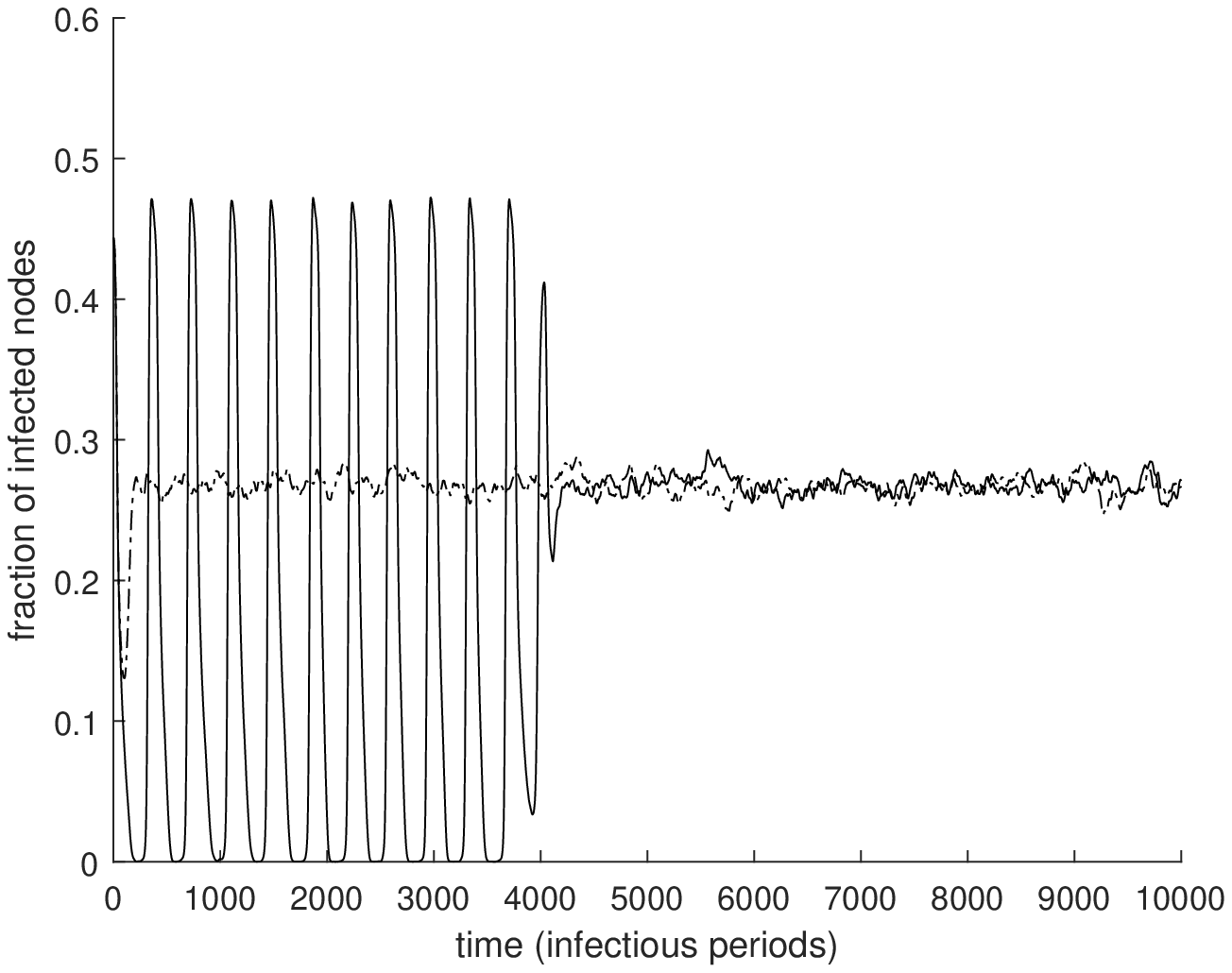}  
\end{subfigure}
\begin{subfigure}{.33\textwidth}
  \centering
  \includegraphics[scale=.4]{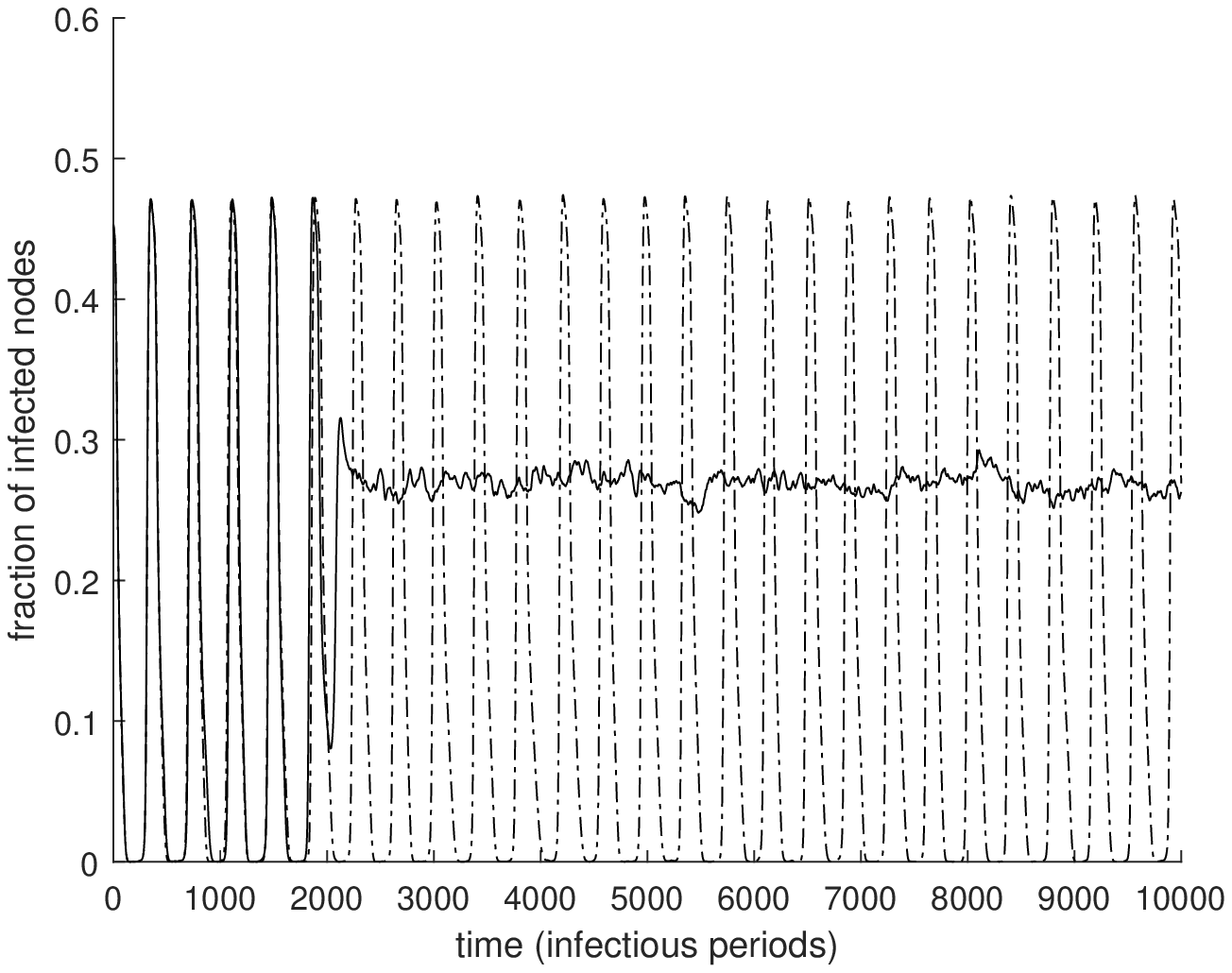}  
\end{subfigure}\\
\begin{subfigure}{.33\textwidth}
  \centering
  \includegraphics[scale=.4]{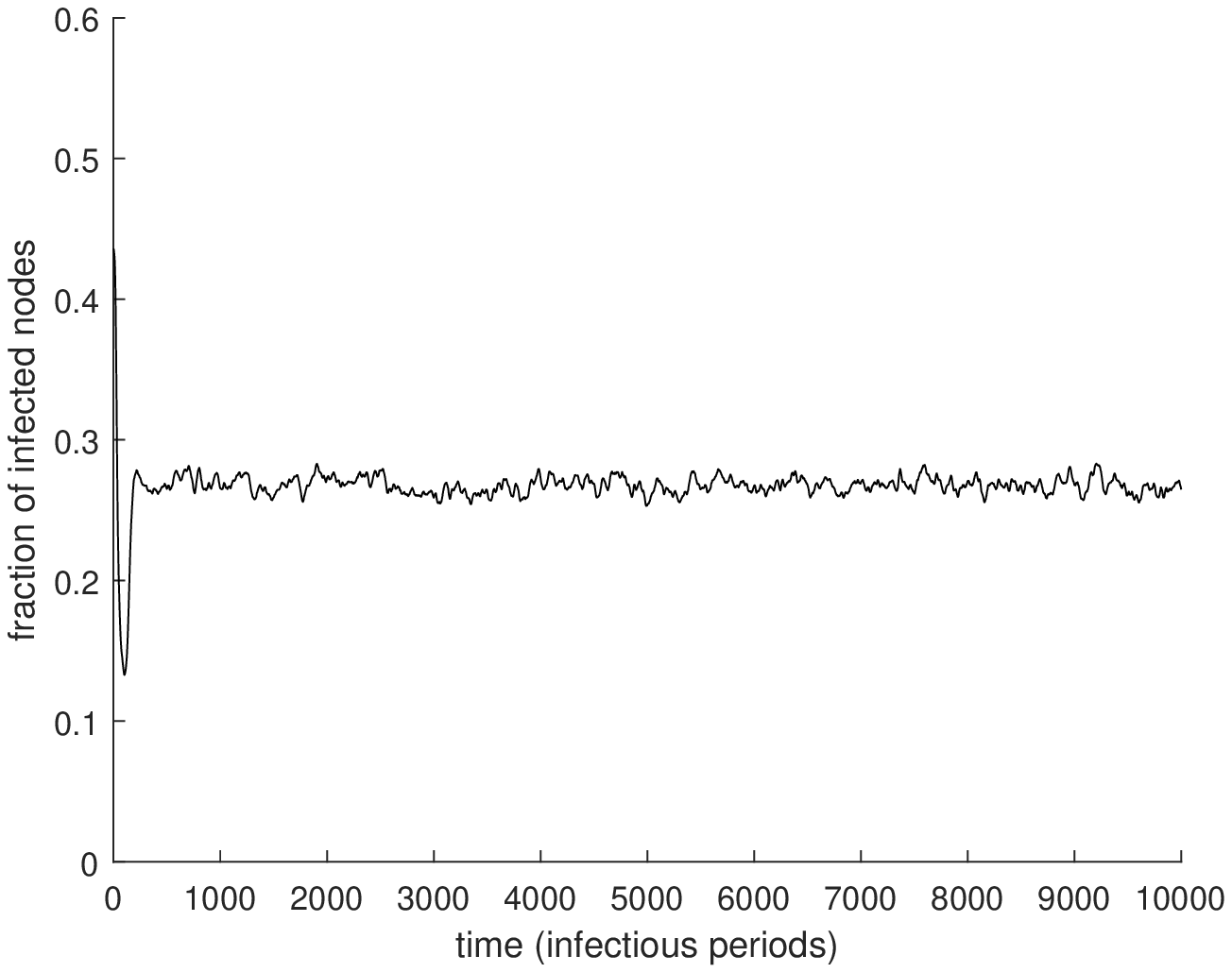}  
\end{subfigure}
\begin{subfigure}{.33\textwidth}
  \centering
  \includegraphics[scale=.4]{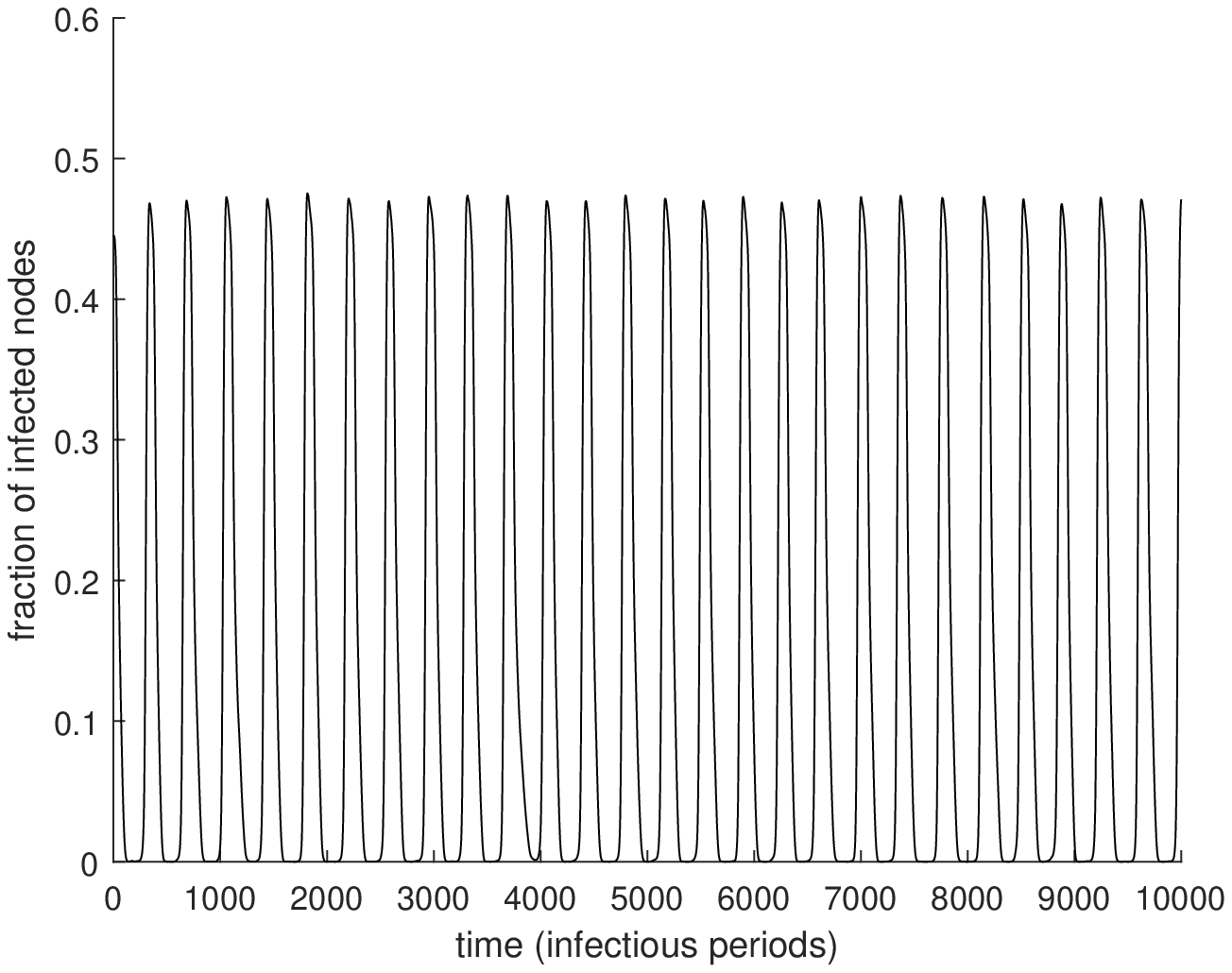}
\end{subfigure} 
\begin{subfigure}{.33\textwidth}
  \centering
  \includegraphics[scale=.4]{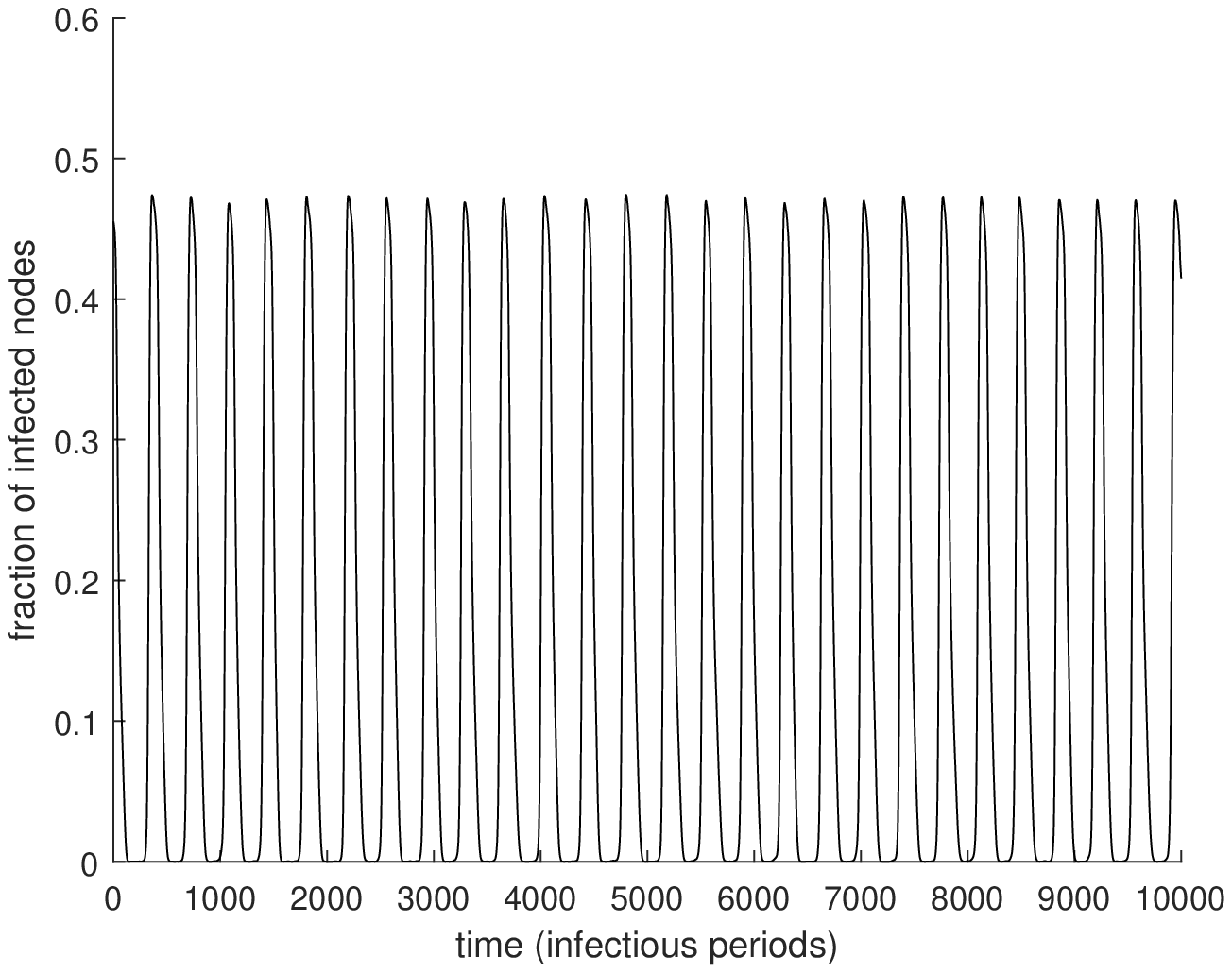}
\end{subfigure} 
\caption{On top, two outputs from GA with parameter $m=17$ and initial conditions $a_0=0.033240$, $i_0=0.437330$, $u_0=0.039240$ (left), $a_0=0.033240$, $i_0=0.447330$, $u_0=0.025900$ (center),  $a_0=0.033240$, $i_0=0.457330$, $u_0=0.012570$ (right). On bottom, a single output from GA with parameter $m=18$ and initial conditions $a_0=0.033240$, $i_0=0.437420$, $u_0=0.039140$ (left), $a_0=0.033240$, $i_0=0.447420$, $u_0=0.025810$ (center),  $a_0=0.033240$, $i_0=0.457420$, $u_0=0.012480$ (right). Parameters of the simulation: $N=10^{5}$.}
\label{fig:experiments_simus_m_gran}
\end{figure}

The previous discussion motivates to choose a more useful representation of the different realizations (see \cite{Aguiar} for a similar approach). To this end, and with the aim of producing an stochastic bifurcation diagram, for each parameter $m$ and each initial condition we represent the maximum amplitude of each individual realization of the experiment in the time interval $[1000,10000]$. For instance, the maximum amplitude of approximately $1000$ different realizations is shown in Figure~\ref{fig:bifurcation_diagram_simus}. We omit the first part of the realization seeking for stationarity of the time series. On the left panel, a transition can be noticed around the parameter value $m=4$, where amplitudes pass from $0.45$ to nearly neglectible (around $0.1$). On the right panel, a similar situation occurs near $m=16$. These two values are close to the bifurcation points $m_{sn1}\approx 3.761$ and $m_{sn2}\approx 16.057$ numerically computed for the system~\eqref{eqn:SAUIS}. However, a clear difference on the abruptness of the transition can be appreciated between the two bifurcation points. On the right-hand panel, the bistability scenario is clearly representated and only a small amount of realizations near $m_{sn2}$ differ from their amplitudes. However, in the left-hand panel the change in the amplitude seems to be more continuous. This happens for two reasons. The first one is, as usual, related to the number of nodes. In Figure~\ref{fig:bifurcation_diagram_simus_zoom} a zoom near $m=3.8$ is given, showing two panels with $N=10^5$ (left) and $N=10^6$ (right). We note that the region of bistability is narrower and closer to $m_{sn1}$ in the right-hand panel. Therefore we can expect a better diagram as $N$ increases. The second reason is dynamical. First, the region of bistability in the parameter space for $m_h<m<m_{sn1}$ is really small compared with $m>m_{sn2}$. Second, the stability of the limit cycles is also very different. For $m_h<m<m_{sn1}$ the unstable limit cycle has very weak repulsion as shown in the left panel of Figure~\ref{fig:lyapunov_exp}. As a consequence, the dynamics near the unstable limit cycle can be misunderstood as stochastic periodic solutions due to the slow decay of the amplitude (even small stochastic perturbations may counter the decay). This is clearly representated in Figure~\ref{fig:retrats_fase}. On the other hand, for $m>m_{sn2}$ the repulsion of the unstable limit cycle is very strong as shown in the right-hand panel of Figure~\ref{fig:lyapunov_exp} and this makes the realizations either stay on the stable limit cycle or rapidly tend to the equilibrium, as shown in Figure~\ref{fig:experiments_simus_m_gran}.

\begin{figure}[ht]
\begin{subfigure}{.49\textwidth}
    \centering
    \includegraphics[scale=0.55]{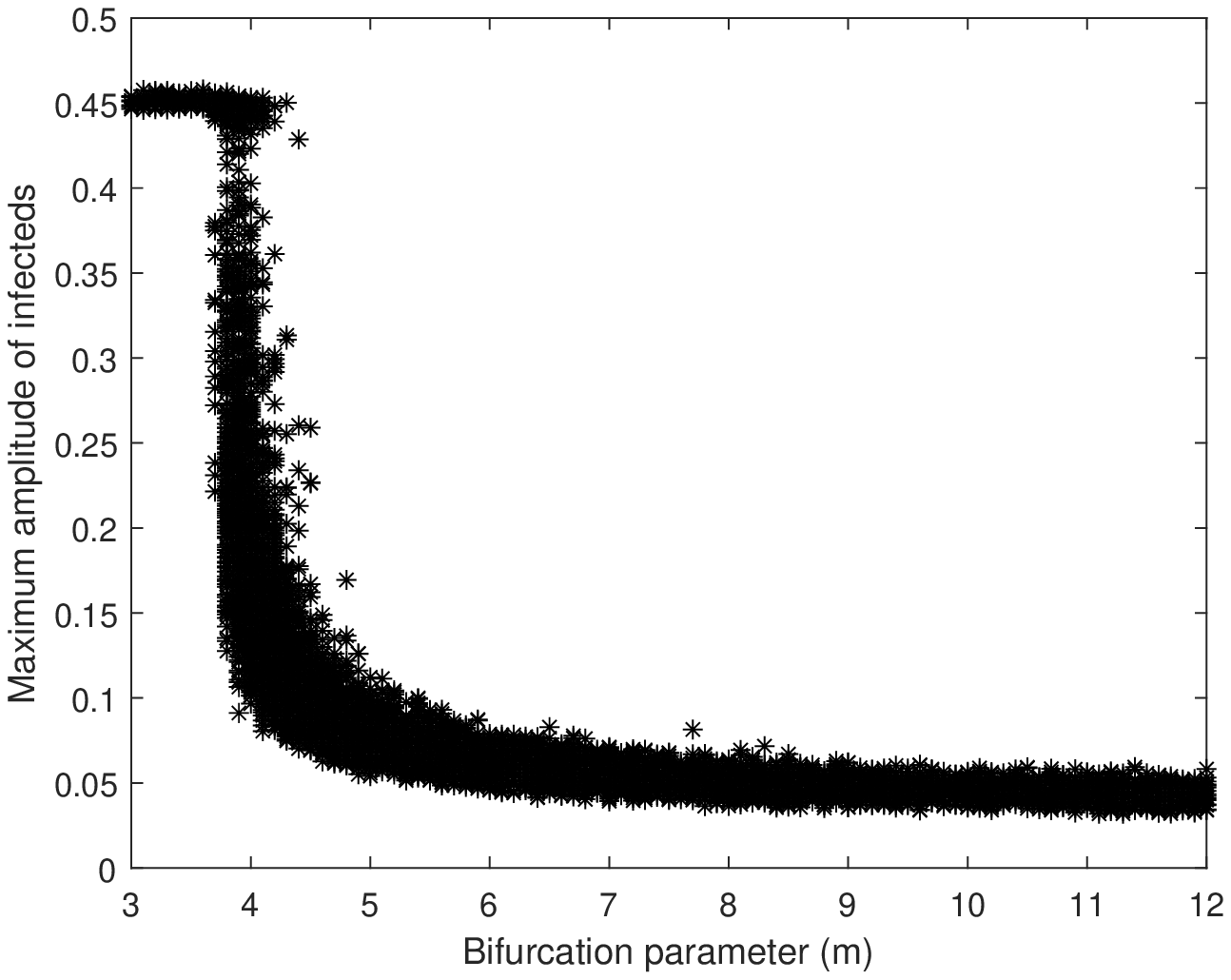}
\end{subfigure}
\begin{subfigure}{.49\textwidth}
    \centering
    \includegraphics[scale=0.55]{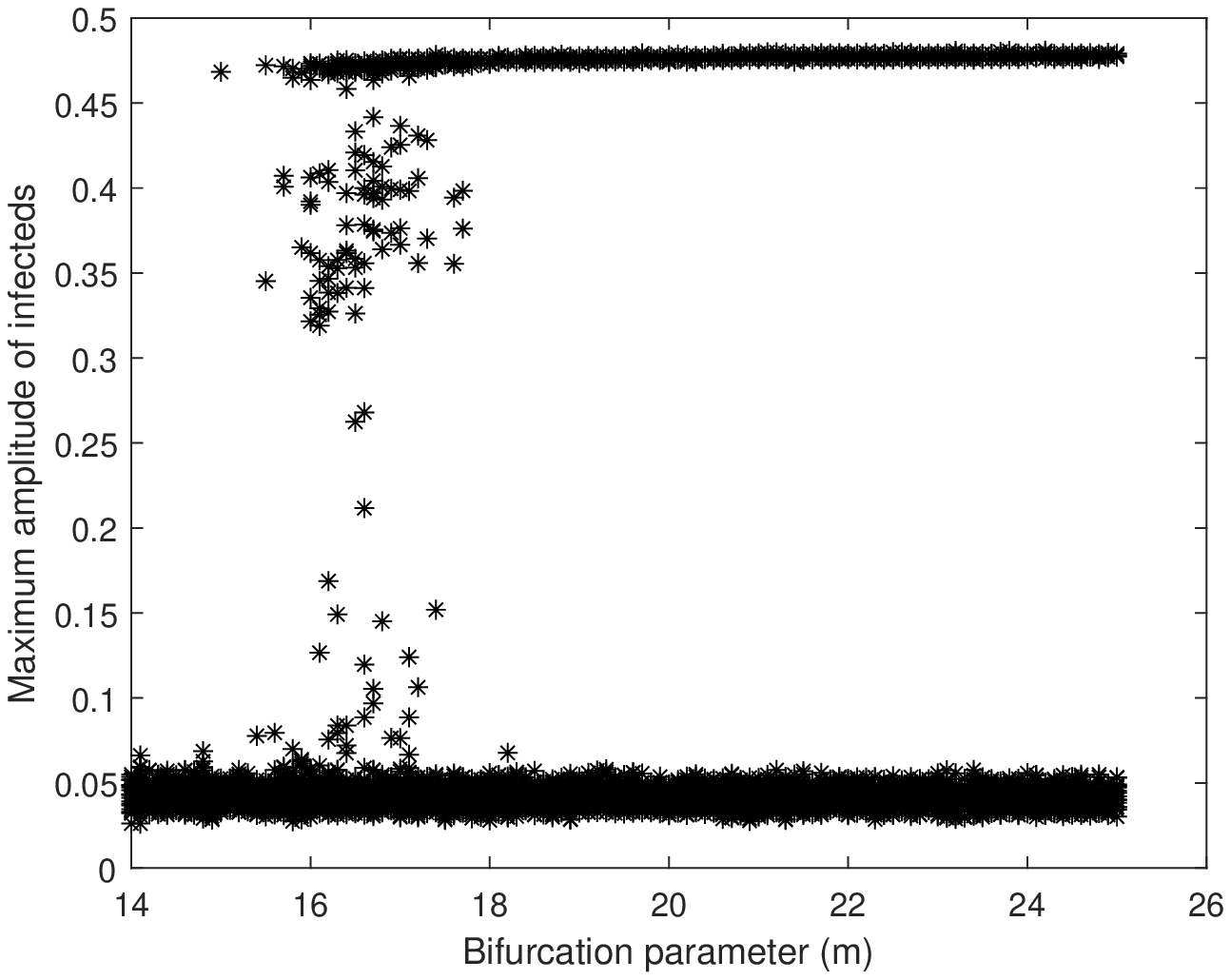}
\end{subfigure}
    \caption{Maximum amplitude of each GA output in the time interval $[1000,10000]$. Parameters of the simulation: $N=10^{5}$.}
    \label{fig:bifurcation_diagram_simus}
\end{figure}

\begin{figure}[ht]
\begin{subfigure}{.49\textwidth}
    \centering
    \includegraphics[scale=0.55]{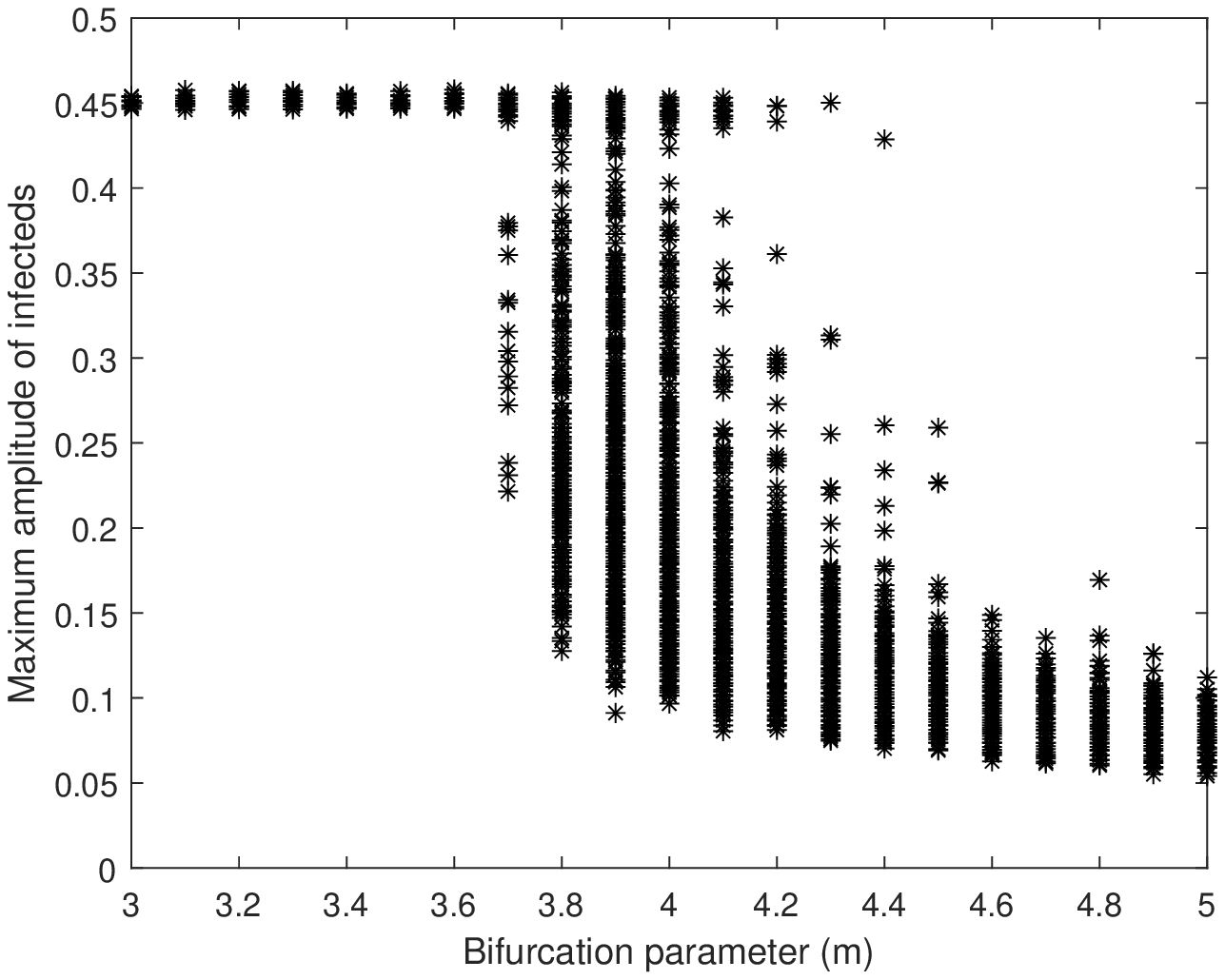}
\end{subfigure}
\begin{subfigure}{.49\textwidth}
    \centering
    \includegraphics[scale=0.55]{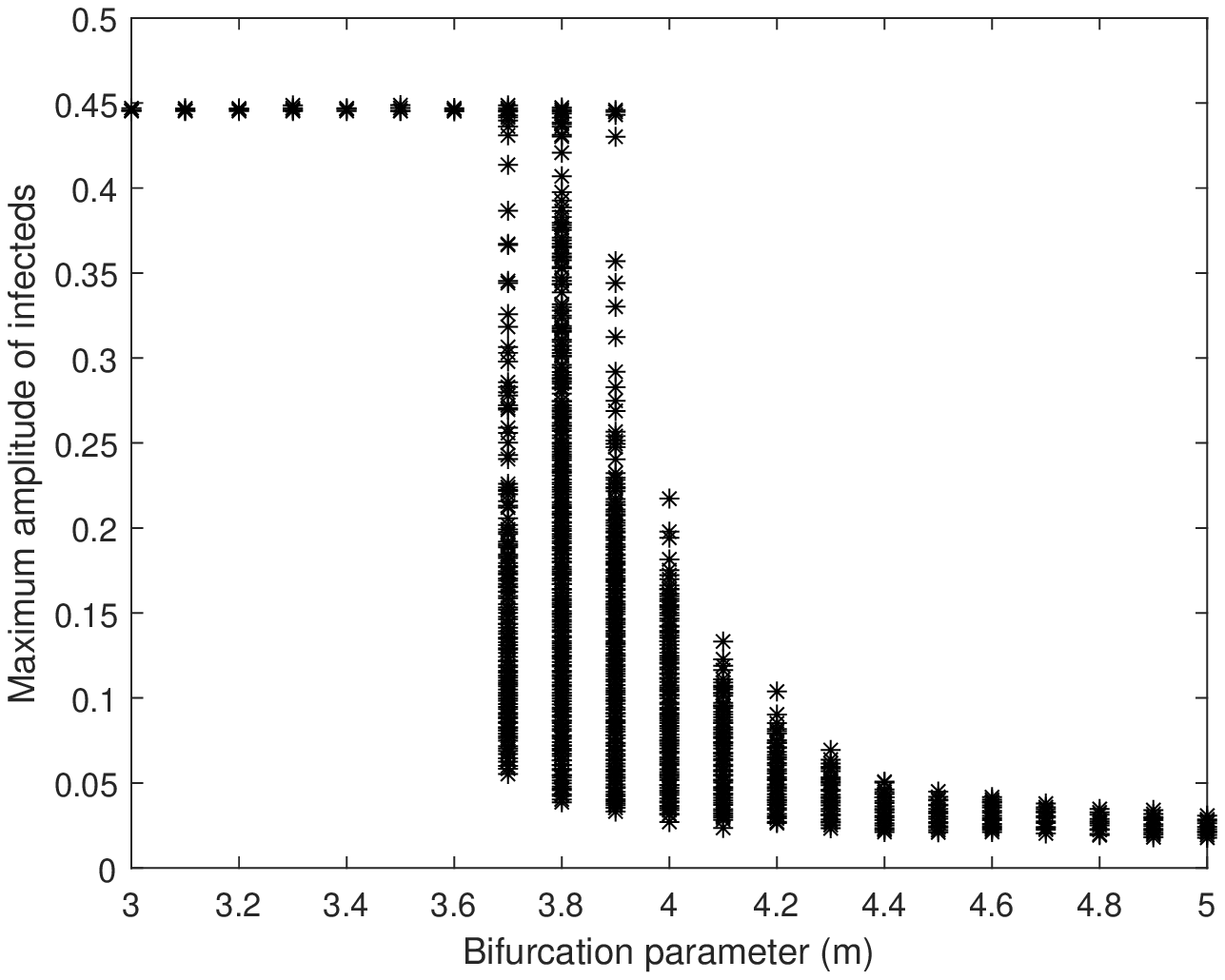}
\end{subfigure}
    \caption{Maximum amplitude of each GA output in the time interval $[1000,10000]$. Parameters of the simulation: $N=10^{5}$ (left). $N=10^{6}$ (right).}
    \label{fig:bifurcation_diagram_simus_zoom}
\end{figure}

\section{Conclusions}

The existence of a high number of infectious cases in a population during an epidemic can modify our individual behaviour and how we relate to others. In turn, behavioural changes modify the  epidemic spread itself. This interplay has long been considered in many papers dealing with classic compartmental models with incidence functions depending on the disease prevalence in a nonlinear way (\cite{Alexander,Liu,Liu2,Ruan}). An alternative approach is based on the addition of new compartments for alerted/responsive individuals that have reduced transmision rates \cite{Funk09,Funk10,Kiss10,Sahneh12a,Szabo}. 

In this paper, we combine both approaches by considering an epidemic model without demography which includes two types of aware individuals who are distinguished by their willingness to convince susceptible individuals to adopt preventive measures. For the model with constant rates, we know that oscillatory solutions can appear as a consequence of a supercritical Hopf bifurcation from the endemic equilibrium \cite{JRS2}. Now, assuming that the rate of alerting decay $\delta_a$ as well as the rate of creation of new unwilling individuals $\nu_a$ decrease by a nonlinear reduction factor $\sigma_m(i) \in (0,1]$ as the prevalence $i$ of the disease increases, we have shown the existence of two scenarios where a bistable configuration with a stable limit cycle and a stable endemic equilibrium occur. Precisely, we assume an abrupt change of both rates when disease prevalence crosses a threshold value $\eta$ (the half-saturation constant). Below this threshold, the value of the rates are close to their maximum values ($\delta_a$ and $\nu_a$), whereas they clearly decrease above it. The sharpness of this change is controlled by a parameter $m$ which determines the slope of $\sigma_m(i)$ at $i=\eta$ ($\sigma'_m(\eta)=-m/(4\eta)$). In both scenarios, the parameters values are in agreement with the sufficient conditions we have obtained for the existence of, at least, one endemic equilibrium.     

Values of $m \gg 1$ can be associated with radical changes in the self-initiated individual behaviour when the prevalence level is close to $\eta$. In this case, $\sigma_m(i) \approx 0$ for $\eta < i \le 1$ which implies almost no decay of awareness and almost no creation of unwilling individuals (only fully aware individuals are created).  For the parameters considered in the paper, a bistable configuration is always the case for $m > 16.06$ ($\sigma'_m(\eta) < -10$) after the occurrence of a saddle-node bifurcation of limit cycles. This configuration is clearly observed in the stochastic simulations of the epidemic process with a very low rate of imported cases due to the strongly repulsive character of the unstable limit cycle lying between the stable one and the endemic euilibrium. For lower values of $m$ (here $m < 4$), the reduction of both rates is not so abrupt and bistability is only present for a narrow range of values ($3.698 < m < 3.761$ with $\sigma'_m(\eta) \approx -2.3$) once a subcritical Hopf bifurcation has occurred. In this case, $\sigma_m(i)$ is clearly positive for $\eta < i \le 1$. For $m < 3.698$, the smoothness of the transition between high and low values of $\sigma_m$ ($\sigma'_m(\eta) \in [-2.3, -0.625]$) as well as the lower reduction of the two rates for $i \approx 1$ make the endemic equilibrium unstable and allow for a stable limit cycle.     

The existence of imported cases (at a rate $\epsilon$) assumed in the present work has also been considered elsewhere (see, for instance, \cite{Aguiar, AlmCon2019}). In addition to its suitability when modelling epidemics in non-isolated populations, it prevents the stochastic extinction of oscillatory epidemics when disease prevalence reaches very low levels. From a deterministic point of view, as long as $\epsilon$ is small enough, the continuous dependence of solutions on parameter values guarantees that the attractors of the model with and without imported cases will be very close to each other (see Figure \ref{fig:bifurcation_diagram}).     

In summary, we have shown the existence of bifurcations of limit cycles in the SAUIS model (without demography) when epidemic spread and awareness transmission are coupled through nonlinear rates that depend on the prevalence level in a population. This has been obtained under a choice of parameters values which assumes a much faster transmission among susceptible individuals of both infections and low level of awareness  than the creation of fully aware individuals. These dynamics are also observed in stochastic simulations of the model assuming a (very low) rate of imported cases with large enough population sizes.

\section*{Acknowledgments}
This work is supported by the grants PID2019-104437GB-I00 and PID2020-118281GB-C31 funded by MCIN/AEI/10.13039/501100011033. D.J. and J.S. are respectively members of the \emph{Consolidated Research Groups} 2017 SGR 1617 and 2017 SGR 01392 of the \emph{Generalitat de Catalunya}. D.R. is a Serra H\'unter Fellow.

\section*{Appendix}

In this Appendix we present a modified version of the Poincaré-Miranda theorem on the plane, which is an extended version of the classical Bolzano's theorem in higher dimension. Up to the authors knowledge, Poincaré-Miranda theorem is not trivially deduced in a planar triangular domain. We credit and thank professor Rafael Ortega for the idea of the proof, which relies on the following result of degree theory that we include for the sake of completeness (see the Appendix on degree theory in \cite{Ortega}.)

Let $\Gamma$ be a Jordan curve in $\mathbb{R}^2$ and let $\Omega$ be the open set enclosed by $\Gamma$. Let $f:\bar\Omega\rightarrow\mathbb{R}^2$ be continuous such that $f(x)\neq 0$ for all $x\in\partial\Omega=\Gamma$. The degree of $f$ in $\Omega$, $\text{deg}(f,\Omega)$, can be computed as the winding number of $f(\Gamma)$ around the origin. 

\begin{thm}\label{thm:degree}
Let $f:\bar\Omega\rightarrow\mathbb{R}^2$ be continuous such that $f(x)\neq 0$ for all $x\in\partial\Omega$. If $\text{deg}(f,\Omega)\neq 0$ then $f$ has a zero in $\Omega$.
\end{thm}

Now we state the version of Poincaré-Miranda theorem on a tringular domain, which can be easily generalised for any Jordan curve with similar assumptions.

\begin{figure}
    \centering
    \includegraphics{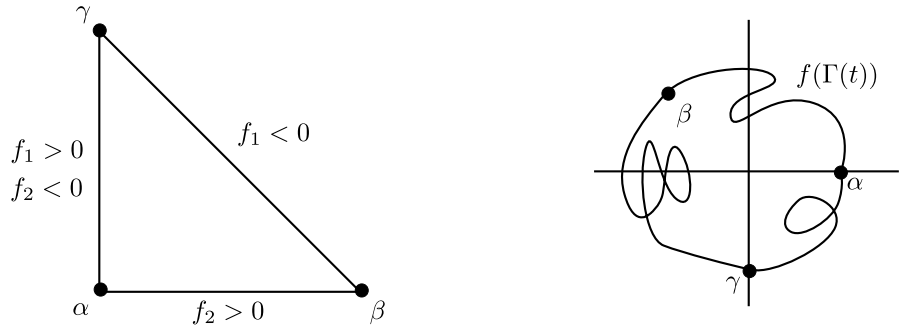}
    \caption{On the left, structure and assumptions of the Poincaré-Miranda theorem for triangles in Theorem~\ref{thm:poincare-miranda}. On the right, an sketch of the behaviour of the function $f(\Gamma(t))$ for $t\in[0,1]$.}
    \label{fig:poincare-miranda}
\end{figure}

\begin{thm}\label{thm:poincare-miranda}
Let $\Omega$ be a triangle and let $f:\mathbb{R}^2\rightarrow\mathbb{R}^2$ be a continuous function, $f(x,y)=(f_1(x,y),f_2(x,y))$. Consider the boundary of $\Omega$ positively oriented and three distinguished points $\alpha$, $\beta$ and $\gamma$ as showed in Figure~\ref{fig:poincare-miranda}. If $f_2>0$ from $\alpha$ to $\beta$, $f_1<0$ from $\beta$ to $\gamma$ and $f_1>0$ and $f_2<0$ from $\gamma$ to $\alpha$, then $f(x,y)$ has at least one zero in $\Omega$.
\end{thm}
\begin{proof}
Let us assume, with aim of reaching contradiction, that $f$ does not vanish on $\Omega$. Let $\Gamma:[0,1]\rightarrow \mathbb{R}^2$ be a curve travelling the boundary of $\Omega$ with $\Gamma(0)=\Gamma(1)=\alpha$. Notice that, by hypothesis, $f$ does not vanish on the boundary of $\Omega$. So, we can consider a continuous argument function $\theta:[0,1]\rightarrow\mathbb{R}$ such that
\[
f(\Gamma(t))=|f(\Gamma(t))|\left(\cos\theta(t),\sin\theta(t)\right),
\]
which is unique up to an additive constant $2\pi$. By the hypotheses of the theorem, $f(\Gamma(t))$ lies on the first and second quadrants when travelling from $\alpha$ to $\beta$, and $f(\beta)$ lies on the second quadrant. Similarly, from $\beta$ to $\gamma$, $f(\Gamma(t))$ lies on the second and third quadrants, finishing on the negative ordinate semi-axis. Finally, from $\gamma$ to $\alpha$, $f(\Gamma(t))$ liess on the fourth quadrant ending on the positive abscissa semi-axis where it started (see Figure~\ref{fig:poincare-miranda}). Therefore, the winding number of the curve $f(\Gamma(t))$ is $\frac{\theta(1)-\theta(0)}{2\pi}=1$. By the definition of the degree, $\text{deg}(f,\Omega,0)=1$ and then, by Theorem~\ref{thm:degree}, $f$ has at least one zero in $\Omega$.
\end{proof}

\end{document}